\documentclass[12pt]{J25ARX}
\pdfoutput=1
\usepackage{titling}
\newcommand{\cO}[1]{}
\usepackage{cite}

\usepackage{natbib}
\newcounter{assum}
\newtheorem{remark}{Remark}
\newtheorem{deff}{Definition}

\usepackage{units}
\usepackage{enumerate}
\usepackage{booktabs}
\usepackage{tabularx}
\usepackage{arydshln}
\usepackage{amssymb}
\usepackage{xcolor}
\usepackage{graphicx}
\usepackage{subfigure}
\usepackage{tikz}
\usepackage{pgfplots}
    \usepackage{caption}
\usetikzlibrary{decorations.pathreplacing,calligraphy}
\pgfplotsset{width=10cm,compat=1.9}
     \newcommand{\jenvelope}{join envelope}
          
\definecolor{green}{HTML}{30AE17}
\usetikzlibrary{graphs}

\title{Signaling with Commitment}
\author{Raphael Boleslavsky   \and Mehdi Shadmehr\thanks{Boleslavsky: Department of Business Economics \& Public Policy, Indiana University (e-mail: \texttt{rabole@iu.edu}); Shadmehr: Department of Public Policy and Department of Economics, University of North Carolina at Chapel Hill (e-mail: \texttt{mshadmeh@gmail.com}). We thank Nageeb Ali, Nemanja Antic, Heski Bar-Isaac, Mike Baye, Dan Bernhardt, Ben Brooks, Steven Callander, Rick Harbaugh, Alexander Jakobsen, Thomas Jungbauer, Navin Kartik, Elliot Lipnowski, John Maxwell, Paula Onuchic, Marilyn Pease, Luciano Pomatto, Doron Ravid, Santanu Roy, Fedor Sandomirskiy, Mark Whitmeyer, Kun Zhang, and seminar audiences at the Kelley School of Business, New York University (Stern), SAET 2023, SEA 2023, and Southern Methodist University. Boleslavsky thanks the Weimer Faculty Fellowship at Indiana University for financial support.}}% }}
\date{\today}

\begin{document}

\begin{titlingpage}

\maketitle

\abstract{
We study the canonical signaling game, endowing the sender with commitment power: before learning the state, sender designs a strategy, which maps the state into a probability distribution over actions. We provide a geometric characterization of the sender's attainable payoffs, described by the topological join of the graphs of the interim payoff functions associated with different sender actions.  We then incorporate payoff irrelevant messages into the game, characterizing the attainable payoffs when sender commits to (i) both messages and actions, and (ii) only messages. By studying the tradeoffs between these model variations, we highlight the power of commitment to actions. We apply our results to the design of adjudication procedures and rating systems.

\vspace{.3in}


\vspace{.1in}

}
\end{titlingpage}
%Word Count $\approx$ 

\let\markeverypar\everypar
 \newtoks\everypar
 \everypar\markeverypar
 \markeverypar{\the\everypar\looseness=-2\relax}
\parskip=2.8pt

\newpage
\baselineskip=20.pt

%\section{Introduction}
\setcounter{page}{1}

\section{Introduction}

Signaling games, including Spence signaling \citep{S1973}, reputation games \citep{kreps/wilson:1982}, communication with lying costs \citep{K2009}, and burning money \citep{Austen-Smith/Banks:2001} arise naturally in various strategic environments, including labor economics, industrial organization \citep{N1974,MR1986,BR1991}, finance \citep{LP1977, R1977,DD1999,D2005}, macroeconomics \citep{Angeletos/Hellwig/Pavan:2006}, organizational economics \citep{Bar-Isaac/Deb:2020}, political economy \citep{B1991,AcemogluEtAl2013}, and cultural economics \citep{S1973b,C1988,Benabou/Tirole:2006}. We study a standard class of signaling games, with the novel feature that the sender can commit to his strategy before learning the state. Depending on the application, such commitment power could arise from the design of institutions, formal contracts, reputation incentives, algorithms, AI, or smart contracts.\footnote{A number of papers provide justification for the assumption that sender can commit to a communication strategy in the Bayesian Persuasion framework, many of which have natural analogs in our setting. \citet{BQ2020} derive conditions under which the desire to maintain influence in future interactions can generate commitment power for the sender. \citet{DPS2022} study how a formal contract allows sender to gain commitment power in the production of information. Moreover, as \citet{SBo15} and \citet{M2020} argue, by choosing the personnel composition of institutions, even autocrats can commit to different policies.} We call this strategic environment \textit{signaling with commitment}.

In particular, consider the standard signaling game described in \citet[pp. 324-5]{FT1991}. There is a sender (leader) and a receiver (follower). There is a state of the world $\omega\in\Omega$, and a common prior,  $\omega\sim \mu_0$ with full support.  The sender observes the state of the world and chooses an action $s\in S$. The receiver observes the sender's action $s$ and chooses an action $a\in A$.  The players' payoffs depend on the state and actions: let sender's payoff be $v(s,a,\omega)$ and receiver's be $u(s,a,\omega)$. The sender's strategy $\Pi$ is a set of probability distributions $\pi(\cdot|\omega)$ over actions $s\in S$, conditional on state $\omega\in\Omega$. Receiver's strategy is also a set of probability distributions $\sigma(\cdot|s)$ over actions $a\in A$, conditional on sender action $s\in S$. We build on the standard game by allowing sender to commit to his strategy $\Pi$ before the state is realized.

In this paper, we study signaling with commitment using the ``belief-based'' approach. In particular, we do not study the sender's strategy directly; instead, we focus on the receiver's posterior beliefs, which are generated by Bayes' rule applied to the sender's strategy. In other words, we treat the sender's strategy $\Pi=\{\pi(s|\omega)\,|\, s\in S,\, \omega\in \Omega\}$ as a signal structure or statistical experiment, which conveys information about the state $\omega$ to the receiver. The sender's realized action $s$ is treated as the realization of the signal or the outcome of the experiment. Thus, the problem of finding the sender's optimal strategy in a signaling game with commitment is replaced with the problem of designing the sender's optimal statistical experiment. In contrast to the information design literature \citep{KG2011, BM19, %K19, 
T2019}, in which the experiment only affects payoffs by revealing information about the state, in our setting the experiment's realization is the sender's action in the signaling game. It therefore directly affects the players' payoffs, beyond its effect on the receiver's belief and action.

We provide a geometric characterization of the sender's optimal strategy and payoff. Applying the belief-based approach, we graph the sender's payoff in the space of posterior beliefs.
Because the sender's payoff depends directly on the signal realization (his action) and on the posterior belief it induces, in the space of beliefs, the sender has a separate payoff function for each action $s\in S$. We show that the set of feasible payoffs for the sender can be characterized geometrically as the \textit{topological join}, ``join'' henceforth, of the graphs of these payoff functions.\footnote{If $A$ and $B$ are subsets of $\mathbb{R} ^n$, then the join of $A$ and $B$ is the set of all line-segments connecting a point in set $A$ to a point in set $B$ \citep{RS82}.} Though it can be represented as a particular convex combination of the sender's payoffs, in general the join is a subset of the convex hull of the sender's payoffs, is not a convex set, and its upper boundary cannot be found via concavification, contrasting with information design.

In signaling with commitment, the receiver's only source of information about the state is the sender's action.  For example, an organization that commits to procedures for adjudicating internal grievances may be barred from releasing additional information about individual cases. Similarly, a financial analyst who commits to a procedure by which he publicly rates assets may face sanctions for revealing additional information privately. More broadly, the sender's interaction with the receiver may be governed by specific rules, institutions, or technologies which prevent the sender from transmitting information beyond his action.

In some environments, such restrictions might be relaxed. To address this possibility, we study a benchmark with ``extended commitment,'' in which sender can provide the receiver with additional information. In particular,  we study commitment power in the cheap talk extension of the signaling game: we incorporate a large space $M$ of payoff-irrelevant messages, and allow sender to commit to a joint distribution of message and action $(m,s)\in M\times S$, conditional on the state.  Thus, commits to transmit information using the action, message, or any combination of the two.

We show that in the extended commitment benchmark, sender can achieve any payoff in the convex hull of the union of the payoff graphs, and the sender's optimal payoff lies on its upper boundary. Furthermore, we show that this is the highest payoff that sender can attain in a signaling game without changing the payoff functions, the prior belief, or the extensive form. By comparing extended commitment and signaling with commitment, we characterize settings in which the extended commitment payoff can be attained \textit{without} transmitting additional messages---commitment to actions alone is sufficient. 

We also show that in order to achieve the extended commitment payoff, commitment to actions is essential---in general, it cannot be achieved by committing to a message protocol alone. To do so, we analyze a ``pre-persuasion'' benchmark in which sender can commit to a message strategy but not actions. Specifically, we augment the signaling game with an initial persuasion stage. In this stage, the uninformed sender designs a statistical experiment that reveals information about the state. After the experiment is realized, the signaling game is played: sender privately learns the state, selects an action, and then receiver observes the action and responds. We construct the sender's highest attainable payoff. In this construction, the sender's payoff graphs must be restricted to include only combinations of beliefs and actions at which a pooling equilibrium exists. Any payoff inside the convex hull of these restricted graphs can be attained with pre-persuasion, and the highest payoff lies on its upper boundary. Because the payoff graphs must be restricted when constructing the convex hull, sender's payoff with pre-persuasion is generally lower than with extended commitment.

Throughout the paper, we illustrate our findings in an example. An organization designs an adjudication procedure that maps valid (invalid) grievances into a probability distribution over resolutions, either dismissing the case or a costly remedy. As mentioned above, confidentiality restrictions prevent the organization from revealing information about the merits of a grievance, beyond what is inferred from the resolution. Outside stakeholders (donors, advertisers, advocacy groups) observe the adjudication procedure and the resolution, and then decide whether to retaliate against the organization (e.g., withhold donations or advertising, stage protests). Stakeholders retaliate if they believe that a grievance has been mishandled, but their power is limited: the organization would rather dismiss the case and incur retaliation than implement a remedy. If the organization cannot commit to a procedure, then it dismisses all grievances and the stakeholders always retaliate. We show that if the organization can commit to its procedures, then it continues to dismiss all grievances when the prior belief is high or low, but for intermediate priors, it remedies valid grievances with positive probability. Removing confidentiality restrictions allows the organization to transmit additional information about the merits of a grievance. With this ability, the organization again dismisses all cases, suggesting a rationale for confidentiality restrictions beyond privacy concerns.

We also study an applications of our results to the design of a rating system, where the designer's payoff depends directly on the nominal rating. This direct preference for nominal ratings introduces a signaling concern which distinguishes the problem from standard information design: the designer's payoff includes terms resembling Spence signaling \citep{S1973b} and ``burned money'' \citep{Austen-Smith/Banks:2001,Kartik:2007}.  A preference for nominal ratings arises naturally in applications. For example, when designing grading policies, schools may internalize professors' distaste for inflating grades of undeserving students. Schools may also internalize the cost of student effort or the value of additional skills  that must be acquired to achieve a high nominal grade \citep{onuchic/Ray:2023}. The possibility of naive receivers, who interpret ratings according to their nominal meanings, also creates a preference for nominal ratings, as in  \citet{KOS2007} and \citet{Inderst/Ottaviani:2012a, Inderst/Ottaviani:2012b}. The existence of such receivers has been documented empirically in financial and labor markets \citep{DeFranco/Lu/Vasvari:2007, Malmendier/Shanthikumar:2007, HHH2023}.

To fix ideas, we focus on an environment in which an analyst designs a rating policy for a financial asset with the goal of increasing investment. Sophisticated investors observe both the rating policy and the realized rating when updating their beliefs; naive investors update as if the rating policy is honest. Thus, a favorable rating leads to greater investment than is warranted given its true information content, and an unfavorable rating leads to less. In addition, we assume that exaggeration is costly: assigning the high rating to a bad asset may  open the analyst to legal sanction or create psychological distress. The analyst's optimal rating policy with commitment  differs significantly from the ratings that arise without commitment. Showing that the optimal rating policy depends on the global shape  of the analyst's payoff, we characterize conditions under which the optimal rating \textit{understates} the asset value, becomes \textit{more} informative when there are more naive investors, and  \textit{inverts} the ratings' nominal meanings---a surprising result that is consistent with some puzzling empirical findings \citep{DeFranco/Lu/Vasvari:2007}. 

With extended commitment, the analyst  provides additional information to sophisticated investors, which is not observed (or processed) by the naive investors. In the US, a number of regulations impose strict restrictions on such private communication, with the intent of protecting investors that do not observe it. We show that such restrictions may bring unintended consequences: the ability to communicate privately with sophisticated investors changes the analyst's public rating policy, and, paradoxically, \textit{increases} the naive investors' payoffs when the prior belief about asset quality is low.

Our paper connects to the growing literature on Bayesian persuasion and information design. In \citet{KG2011}, sender and receiver play a cheap talk game, and sender commits to his message strategy---or equivalently, to a statistical experiment that generates information about the state. In this characterization, the sender's optimal payoff is the concave envelope of the sender's payoff graph in the space of posterior beliefs. This belief-based approach has been used by \citet{AC2016} to study political persuasion, by \citet[2018]{BC2015} and \citet{AK2020} to study competitive information production, and by \citet{ZZ2016} to study contest design. In concurrent work, \citet{KLT2024} develop  a belief-based approach for the study of signaling games without commitment. A number of papers use the belief-based approach to study costly information production or rational inattention \citep{S2003, CD2015}. In this literature, the cost of an information structure is posterior-separable \citep{Denti:2022}---it can be expressed as an expectation of a function of the realized posterior. \citet{PST2023} provide an axiomatic foundation for a class of such cost functionals. In contrast to these literatures, signaling with commitment can be viewed as an information design problem in which the realization of the experiment is directly payoff-relevant, beyond its effect on beliefs.

\subsection{An Example}\label{EX1} To communicate the key ideas, we begin with an example. An organization (sender, he) designs procedures or formal rules for addressing grievances. A grievance is either valid (e.g., a true violation of Title IX), which we denote $\omega=v$, or invalid, $\omega=f$ (e.g., a mistaken or false accusation). Both types of grievance arise exogenously within the organization at different rates. Consequently, a new grievance is believed to be invalid with probability $\mu_0$. After learning the details of the grievance and determining its type, the organization has two available actions, $S=\{a,d\}$. By selecting $a$,  the organization addresses the grievance and implements a costly remedy; by selecting $d$, it dismisses the grievance without redress. In either case, stakeholders---including potential donors, customers, advertisers, partners, or advocacy groups---observe the organization's decision and update their beliefs, $\mu_s=\Pr(\omega=f|s)$ for $s\in\{a,d\}$. Moreover, confidentiality restrictions prohibit the organization from communicating about the details of individual grievance cases. Thus, the organization's action is the stakeholders' only signal about the merits of the case. The stakeholders, as a unitary actor (receiver, she), then decide whether to retaliate against the organization, for example by withholding donations, boycotting, dissolving partnership, or protesting. In particular, stakeholders retaliate when they believe that the organization is likely to have mishandled the grievance. Thus, if the organization dismisses the grievance, stakeholders retaliate when they believe that the grievance is likely valid, $\mu_d<\theta_d$; if the organization addresses the grievance, then the stakeholders retaliate when they believe the case is likely to be invalid, $\mu_a>\theta_a$ (where $\theta_a$ and $\theta_d$ are exogenous). For concreteness, we focus on $\theta_a<\theta_d$, but the following analysis holds for all $\theta_a$, including $\theta_a>1$.

Both redress and retaliation are costly for the organization. The organization's payoff from redress ($s=a$) is 0, and its payoff from dismissal ($s=d$) is 1, provided that the stakeholders do not retaliate. Retaliation by stakeholders imposes a cost of $l\in(0,1)$ on the organization. With $l<1$, stakeholders have limited sway over the organization: the organization would rather dismiss a grievance and face the stakeholders' punishment than address it.

 The organization designs and commits to formal procedures for addressing grievances. These procedures  determine the probability $\pi(s|\omega)$ that action $s\in\{a,d\}$ is taken in state $\omega\in\{v,f\}$.
 Consequently, they determine the stakeholders' posterior belief induced by each realized action, $\mu_s\equiv\Pr(\omega=f|s)$, $s\in\{a,d\}$, and the unconditional probability $\tau(\mu_s)$ of each posterior belief.\footnote{In particular, $\tau(\mu_s)\equiv\mu_0\pi(s|f)+(1-\mu_0)\pi(s|v)$ and $\mu_s=\Pr(\omega=f|s)=\mu_0\pi(s|f)/\tau(\mu_s)$.} The law of iterated expectations implies $E_\tau[\mu_s]=\mu_0$. Furthermore, any combination of beliefs $\{\mu_a,\mu_d\}$ and probabilities $\{\tau(\mu_a),\tau(\mu_d)\}$ that satisfies the law of iterated expectations arises from some adjudication procedure $\pi(\cdot|\cdot)$.
 
 % (known as Bayes plausibility condition). %---weighted by the probabilities $\tau(\mu_s)$, the convex combination of $\mu_d$ and $\mu_a$ is the prior $\mu_0$.

 \begin{figure}
\begin{minipage}{0.5\textwidth}
\begin{tikzpicture}[scale=0.7]

%AXES
\draw[line width=0.5pt] (0,-6.5) -- (0,4); 
\fill (8,-2) node[right] {\footnotesize{$\mu$}};
\draw[line width=0.5pt] (2,-2) -- (8,-2); 
% \draw[line width=0.5pt] (0,0) -- (8,0); 

%AXES LABELS
% \fill (8,0) node[right] {\footnotesize{$\mu$}};

%$\widehat{v}(\mu,a)$

\draw[line width=1pt,blue] (0,0)--(4,0);
\draw[line width=1pt,blue] (4,4)--(8,4);
\draw[line width=0.5pt,blue,dashed] (4,0)--(4,4);
\fill (6,4.5) node[right] {\footnotesize{$\widehat{v}(\cdot,d)$}};
\fill (4,-2) node[below] {\footnotesize{$\theta_d$}};
\draw[line width=0.5pt,black,dotted] (4,-2)--(4,0);
\fill (0,-2) node[left] {\footnotesize{$0$}};
\fill (0,0) node[left] {\footnotesize{$1-l$}};
\draw[line width=0.5pt,black,dotted] (4,4)--(0,4);
\fill (0,4) node[left] {\footnotesize{$1$}};
\fill (0,-6) node[left] {\footnotesize{$-l$}};
\draw[line width=0.5pt,black,dotted] (0,-6)--(2,-6);
\fill (2.5,-2) node[below] {\footnotesize{$\theta_a$}};

% \draw[line width=1pt,violet] (4,4)--(8/6,0);
% \draw[line width=1pt,violet] (4,4)--(8,4);
% \draw[line width=1pt,violet] (0,0)--(8/6,0);

% \fill (8/6+0.15,0) node[below] {\footnotesize{$\widehat\mu$}};

%$\widehat{v}(\mu,m)$
% \draw[line width=0.5pt,black,dashed] (0,-2)--(8/6,0);
\draw[line width=1pt,red] (0,-2)--(2,-2);
\draw[line width=1pt,red] (2,-6)--(8,-6);
\draw[line width=0.5pt,red,dashed] (2,-6)--(2,-2);
 \fill (6,-5.5) node[right] {\footnotesize{$\widehat{v}(\cdot,a)$}};

 % \draw[line width=0.5pt,black,dashed] (2,-6)--(0,0);

% \draw[line width=0.5pt,black,dashed] (0,-2)--(8/6,0);
%Join 2
%\draw[line width=0.5pt,black,dashed] (8/3,2)--(0.73,0.53);
%\draw[line width=0.5pt,black,dashed] (0.71,0.15)--(8/3,2);
%ENVELOPE
%\draw[line width=1pt,violet] (0,0)--(0.73,0.53);

% \draw[line width=1pt,violet] (4,4)--(8/6,0);
% \draw[line width=1pt,violet] (4,4)--(8,4);
% \draw[line width=1pt,violet] (0,0)--(8/6,0);
% \draw[line width=1pt,green] (0,0)--(4,4);
%\draw[line width=1pt,violet] (4,4.07)--(8,4.07);
%\draw[line width=0.5pt,black, dashed] (0.73,0.53)--(0.73,0);
% \filldraw[color=red, fill=white] (0,-2) circle (3pt);
% \filldraw[color=red, fill=red] (0,-2) circle (3pt);
\filldraw[color=blue, fill=blue] (4,4) circle (3pt);
% \filldraw[color=blue, fill=white] (0,0) circle (3pt);

\filldraw[color=red, fill=red] (2,-2) circle (3pt);

% \filldraw[color=red, fill=white] (0,-2) circle (3pt);

% \fill[violet,opacity=0.05] (0,-2)--(0,0)--(8/6,0)--(4,4)--(8,4)--(8,-6)--(2,-6)--(0.7,-2)--(0,-2);

\end{tikzpicture}
\end{minipage}
 \begin{minipage}{0.5\textwidth}
\begin{tikzpicture}[scale=0.7]

%AXES
\draw[line width=0.5pt] (0,-6.5) -- (0,4); 
\fill (8,-2) node[right] {\footnotesize{$\mu$}};
\draw[line width=0.5pt] (2,-2) -- (8,-2); 
% \draw[line width=0.5pt] (0,0) -- (8,0); 

%AXES LABELS
% \fill (8,0) node[right] {\footnotesize{$\mu$}};

%$\widehat{v}(\mu,a)$

\draw[line width=1pt,blue] (0,0)--(4,0);
\draw[line width=1pt,blue] (4,4)--(8,4);
\filldraw[color=blue, fill=white] (4,4) circle (3pt);
\draw[line width=0.5pt,blue,dashed] (4,0)--(4,4);
\fill (6,4.5) node[right] {\footnotesize{$\widehat{v}(\cdot,d)$}};
\fill (4,-2) node[below] {\footnotesize{$\theta_d$}};
\draw[line width=0.5pt,black,dotted] (4,-2)--(4,0);
\fill (0,-2) node[left] {\footnotesize{$\mu_a^*$}};
\fill (4,4) node[above] {\footnotesize{$\mu_d^*$}};

\fill (0,0) node[left] {\footnotesize{$1-l$}};
\draw[line width=0.5pt,black,dotted] (4,4)--(0,4);
\fill (0,4) node[left] {\footnotesize{$1$}};
\fill (0,-6) node[left] {\footnotesize{$-l$}};
\draw[line width=0.5pt,black,dotted] (0,-6)--(2,-6);
\fill (2.5,-2) node[below] {\footnotesize{$\theta_a$}};

\draw[line width=1pt,violet] (4,4)--(8/6,0);
\draw[line width=1pt,violet] (4,4)--(8,4);
\draw[line width=1pt,violet] (0,0)--(8/6,0);

% \fill (8/6+0.15,0) node[below] {\footnotesize{$\widehat\mu$}};

%$\widehat{v}(\mu,m)$
\draw[line width=0.5pt,black,dashed] (0,-2)--(8/6,0);
\draw[line width=1pt,red] (0,-2)--(2,-2);
\draw[line width=1pt,red] (2,-6)--(8,-6);
\draw[line width=0.5pt,red,dashed] (2,-6)--(2,-2);
 \fill (6,-5.5) node[right] {\footnotesize{$\widehat{v}(\cdot,a)$}};

 \draw[line width=0.5pt,black,dashed] (2,-6)--(0,0);

% \draw[line width=0.5pt,black,dashed] (0,-2)--(8/6,0);
%Join 2
%\draw[line width=0.5pt,black,dashed] (8/3,2)--(0.73,0.53);
%\draw[line width=0.5pt,black,dashed] (0.71,0.15)--(8/3,2);
%ENVELOPE
%\draw[line width=1pt,violet] (0,0)--(0.73,0.53);

% \draw[line width=1pt,violet] (4,4)--(8/6,0);
% \draw[line width=1pt,violet] (4,4)--(8,4);
% \draw[line width=1pt,violet] (0,0)--(8/6,0);
% \draw[line width=1pt,green] (0,0)--(4,4);
%\draw[line width=1pt,violet] (4,4.07)--(8,4.07);
%\draw[line width=0.5pt,black, dashed] (0.73,0.53)--(0.73,0);
\filldraw[color=red, fill=white] (0,-2) circle (3pt);
\filldraw[color=red, fill=white] (0,-2) circle (3pt);
\filldraw[color=blue, fill=white] (4,4) circle (3pt);
\filldraw[color=blue, fill=white] (0,0) circle (3pt);

\filldraw[color=red, fill=white] (0,-2) circle (3pt);
\filldraw[color=red, fill=red] (2,-2) circle (3pt);
\filldraw[color=red, fill=white] (2,-6) circle (3pt);

\fill[violet,opacity=0.05] (0,-2)--(0,0)--(8/6,0)--(4,4)--(8,4)--(8,-6)--(2,-6)--(0.7,-2)--(0,-2);
\fill (8/6,-2) node[below] {\footnotesize{$\widehat\mu$}};
 \draw[line width=0.5pt,black,dotted] (8/6,-2)--(8/6,0);
\end{tikzpicture}
\end{minipage}
\caption{\label{Fig:Adjud} Motivating Example.}
\end{figure}
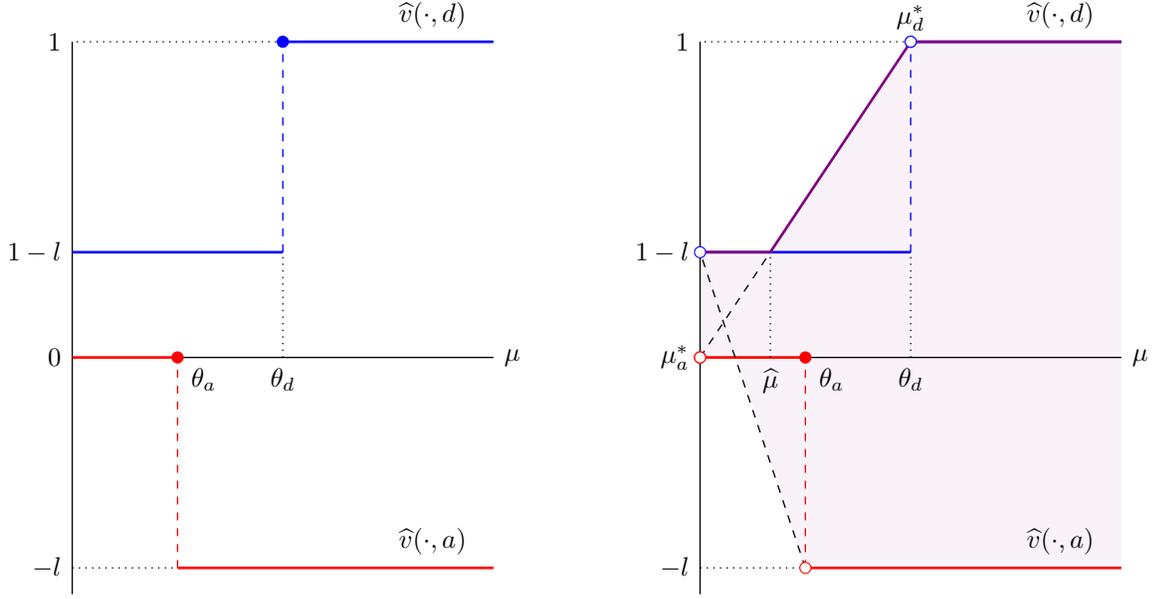

Building on these observations, the organization's optimal commitment can be analyzed graphically using Figure \ref{Fig:Adjud}. On the horizontal axis, we have $\mu$, the stakeholder's posterior belief that the grievance is invalid after observing the organization's action. The blue step function is the graph of $\widehat{v}(\mu,d)=1-\mathcal{I}(\mu<\theta_d)l$, which is the organization's expected payoff if it dismisses the case $(s=d)$ and the stakeholder's updated belief is $\mu$.\footnote{The organization expects 1 if the stakeholders do not retaliate, (i.e., if $\mu\geq \theta_d$) and $1-l$ otherwise.} Similarly, the red step function, $\widehat{v}(\mu,a)=-\mathcal I(\mu>\theta_a)l$, is the organization's expected payoff, if it addresses the case $(s=a)$ and the stakeholder's updated belief is $\mu$.\footnote{Here, the organization expects 0 if there is no retaliation, $\mu\leq\theta_a$ and $-l$ otherwise.}  

As described above, the organization's strategy induces a distribution over posteriors $\{\mu_a,\mu_d\}$, with a mean equal to the prior $\mu_0$. The organization's ex ante expected payoff of such a strategy, is then the expectation of $\widehat{v}(\mu_s,s)$ according to the same distribution over posteriors. Graphically, this payoff can be found by drawing a line segment connecting the point on the red curve above or below $\mu_a$ with the point on the blue curve above or below $\mu_d$, and then evaluating the height of the line segment at $\mu_0$. We then vary the organization's strategy over all possible $\{\mu_a,\mu_d\}$ to increase this height as much as possible, as illustrated in the right panel. Carrying out this procedure, we find that the largest payoff that the organization can attain with prior $\mu_0$ is the height of the purple curve above $\mu_0$. Furthermore, at any prior the shaded purple region represents all payoffs that the organization could attain by varying its procedures.

As Figure \ref{Fig:Adjud} illustrates, when the prior is high or low, $\mu_0\geq\theta_d$ or $\mu_0\leq \widehat\mu$, the organization always dismisses all grievances. In contrast, 
for moderate priors, $\mu_0\in(\widehat{\mu},\theta_d)$, the organization sometimes commits to address a valid grievance with positive probability, despite the fact that the worst payoff from dismissing a grievance is larger than the best payoff from addressing it---the highest value on the red curve is strictly below the lowest value on the blue one. This is evident in Figure \ref{Fig:Adjud}, where moderate priors are optimally spread into posteriors $\{\mu^*_a=0,\mu^*_d=\theta_d\}$.

Intuitively, if valid grievances are too unlikely under the prior, $\mu_0\geq \theta_d$, then the organization does not need to reveal information in order to deter retaliation, and thus dismissing all grievances is optimal. In contrast, for $\mu_0\leq\theta_d$, the organization has an incentive to reveal information in order to deter retaliation. However, in order to increase $\mu_d$, the organization must commit to address valid grievances with positive probability. Because addressing a grievance is costly, deterring retaliation in this way is only worthwhile if valid grievances are relatively unlikely, $\mu_0\in(\widehat{\mu},\theta_d)$. If valid grievances are too likely under the prior, $\mu_0<\widehat{\mu}$, it is better in expectation to dismiss all grievances and incur the wrath of stakeholders rather than paying the costs of redress. 

In this example, the organization's ability to commit to its procedures expands stakeholders influence over the organization. Without commitment the organization dismisses all grievances (recall $1-l>0$). With commitment, however, valid grievances are addressed with positive probability at moderate priors.

One key feature of this characterization is worth highlighting: the upper envelope representing the largest attainable payoff (in purple) is \textit{not} the concave envelope of the organization's payoffs, as it would be in Bayesian Persuasion---indeed it is not even concave. In this example, the concave envelope connects two points on the blue curve (evidently $(0,0)$ and $(\theta_d,1)$). Because the only source of information for the stakeholders is the organization's action, and different actions correspond to different payoff functions, when constructing the line segments that allow us to evaluate the organization's payoff we can only connect a point on the blue curve to a point on the red curve.
If a payoff constructed by connecting two points on the blue curve were feasible, then dismissing a grievance would generate two different posterior beliefs for the stakeholder. This is possible only if some  additional source of information exists, beyond the organization's action. We consider an extension with this feature in Section \ref{ACREC}. As we will see, when the stakeholders have limited sway ($l<1$), relaxing privacy restrictions  and allowing the organization to provide additional information to stakeholders leads to the dismissal of all grievances.

\section{Theory}

 \paragraph{Preliminaries.} To keep the exposition as straightforward as possible, we focus on a finite state space $|\Omega|<\infty$ and sender action space $|S|<\infty$. Furthermore, $A$ is a compact metric space, and both sender and receiver payoffs are continuous in $a$.\footnote{$|A|<\infty$ automatically satisfies these conditions.} A belief $\mu\in\Delta(\Omega)$ is a probability distribution over states, and $\mu(\omega)$ is the probability of state $\omega$ under belief $\mu$. By $\tau\in\Delta(\Delta(\Omega))$ we denote a probability distribution over beliefs, and $\tau(\mu)$ is the probability of belief $\mu$ in distribution $\tau$. Players have a common prior $\mu_0$ with full support on $\Omega$.

\paragraph{Belief Systems and Sender Strategies.}  A belief system $\mathcal{B}\equiv\{\mu_s,\tau(\mu_s)\}_{s\in S}$, is a set of beliefs and associated probabilities, $\tau(\mu_s)$, indexed by the sender's set of available actions. The probabilities in a belief systems must be exhaustive, i.e., $\sum_{s\in S}\tau(\mu_s)=1$, but $\tau(\mu_s)=0$ is allowed. Following \citet{KG2011}, we say that sender strategy $\Pi=\{\pi(s|\omega)\}_{(s,\omega)\in S\times\Omega}$ induces belief system $\mathcal{B}$ if and only if
\begin{align*}
	\tau(\mu_s)=\sum_{\omega\in\Omega}\mu_0(\omega)\pi(s|\omega),\quad\quad \mu_s(\omega)=\frac{\mu_0(\omega)\pi(s|\omega)}{\tau(\mu_s)}.
\end{align*}
If strategy $\Pi$ induces belief system $\mathcal{B}$, then $\mu_s$ is  the posterior belief associated with sender's action $s\in S$, and the probability $\tau(\mu_s)$ is the probability that action $s$ is selected under the sender's strategy. In other words, $\tau(\mu_s)$ is  the probability that the realization of the posterior belief is $\mu_s$. 

\begin{remark}
If certain actions in $S$ are never chosen, i.e. $\tau(\mu_s)=0$, then the associated posterior is undefined. In this case, we allow $\mu_s$ to take on any value. Such realizations of the posterior have no effect on the sender's expected payoff, and, given the sender's commitment power, they play no role in the analysis.	
\end{remark}

As is well-known, a belief system induced by a sender strategy must be Bayes-Plausible: $E_{\tau}[\mu]=\sum_{s\in S}\tau(\mu_s)\mu_s=\mu_0$. Following the argument of \citet{KG2011}, it is straightforward to show that any Bayes-Plausible belief system can  be induced by a corresponding sender strategy. Therefore, rather than designing the sender's strategy, we can consider the sender designing a belief system directly, with the constraint that it is Bayes-Plausible.

\begin{prop}\label{BS1} A belief system $\mathcal{B}$ is Bayes-Plausible if and only if it is induced by some sender strategy $\Pi$.
\end{prop}

\begin{remark}
For Proposition \ref{BS1}, it is important that all sender actions $s\in S$ are available in all states, $\omega\in \Omega$. If sender cannot select certain actions in certain states, then some Bayes-Plausible belief systems cannot be induced by a sender strategy. While important to the analysis, the assumption that all actions are available in all states is without loss of generality. To analyze a setting in which action $s$ is unavailable in $\omega$, one simply imposes a large cost on the sender for such a choice. For a given prior distribution $\mu_0$, a sufficiently large penalty ensures that selecting $s$ in state $\omega$ with positive probability is suboptimal.

\end{remark}

\paragraph{Receiver's Response.} The receiver selects an action after observing the sender's action $s$, which is associated with posterior belief $\mu_s$. The receiver's optimal action(s) solve
\begin{align*}
	\widehat{a}(\mu_s,s)=\arg\max_{a\in A} E_{\mu_s}[u(s,a,\omega)].
\end{align*}
In the environment we study, for each $(\mu_s,s)$ the receiver's problem has at least one solution. Furthermore, whenever the receiver's problem has multiple solutions, the sender's expected payoff  is maximized by at least one of them. We focus on an equilibrium in which (one of) the sender's preferred action(s) is chosen in this situation.\footnote{In particular, among the receiver's optimal actions $a\in\arg\max E_{\mu_s}[u(s,a,\omega)]$, we select one that maximizes sender's expected payoff, $E_{\mu_s}[v(s,a,\omega)]$.}

\paragraph{Sender's Problem.} Equilibrium with sender commitment requires that the belief system (and associated sender strategy) designed by the sender maximize the sender's ex ante payoff. To formulate the sender's problem,  first consider the sender's utility at the interim stage, after the sender's  action has been realized, but before the receiver responds. At this stage, $\mu_s$ is the belief about the state and $\widehat{a}(\mu_s,s)$ is the receiver's response. The sender's interim utility is
\begin{align*}
	\widehat{v}(\mu_s,s)=E_{\mu_s}[v(s,\widehat{a}(\mu_s,s),\omega)].
\end{align*}
From this expression, we see that the signal realization (sender action) affects the sender's interim payoff in two ways, beyond its effect on the posterior belief. Keeping the posterior belief $\mu_s$ constant, changes in $s$ can affect the interim utility directly through the sender's payoff (reflected in the first argument of $v(\cdot)$), or through its effect on the receiver's response (reflected in the second argument of $\widehat{a}(\cdot,\cdot)$).

The sender's problem is therefore to design a belief system, $\mathcal{B}=\{\mu_s,\tau(\mu_s)\}_{s\in S}$ in order to maximize the sender's ex ante expected payoff, subject to the constraint that the belief system is Bayes-Plausible,
\begin{align*}
(\text{Sender's Problem})\quad&\max_{\mathcal{B}} E_\tau[\hat v(\mu_s,s)]\quad
\text{subject to}\quad\ \sum_{s\in S}\ \mu_s \tau(\mu_s)=\mu_0.
\end{align*}
The geometric characterization of the solution is based on the \textit{topological join} (``join'') of sets \citep{RS82}.

\begin{deff}(Join of Sets). If $X_i\subset\mathbb{R}^n$ for $i=1,...,k$, then their join is \begin{align*}\text{join}(X_1,...,X_k)\equiv \Big\{\sum_{i=1}^k\lambda_ix_i\,|\,x_i\in X_i,\,\lambda_i\geq 0,\, \sum_{i=1}^k\lambda_i=1\Big\}.\end{align*}
\end{deff}
The join of the sets $(X_1,...,X_k)$ is  generated by forming all possible convex combinations, using at most one point from each set $X_i$. For two sets $(X_1,X_2)$, it is the union of the sets themselves, along with all line segments connecting a point in $X_1$ to a point in $X_2$. With more than two sets, we can imagine first joining $X_1$ and $X_2$, then joining $X_3$ to the resulting set, and so on. In Figure \ref{Fig:Adjud}, the join of the graphs of the sender's payoff functions $\widehat{v}(\cdot,a)$ and $\widehat{v}(\cdot,d)$ is shaded purple.

In general, the $join(X_1,...,X_k)$ is a subset of the convex hull of the union  $\bigcup_{i=1}^kX_k$, which we denote by $con(X_1,...,X_k)$. Indeed, the convex hull is formed by taking all possible convex combinations of points in $\bigcup_{i=1}^k X_i$, without the requirement that the convex combination includes at most one point from each set. These two coincide in the special case where the convex hull happens to satisfy this additional requirement. Furthermore, unlike the convex hull, the join is not necessarily a convex set.  We return to the connection between the convex hull of the union and the join in Section \ref{ACREC}.

 % {\color{red} what constraint?}

\begin{deff}(Join Envelope). If $X_i\subset\Delta(\Omega)\times \mathbb{R}$ for $i=1,...,k$, then their \jenvelope{} is a function of $\mu\in\Delta(\Omega)$ \begin{align*}V^{jo}(\mu|X_1,...,X_k)\equiv \sup\Big\{z\,|\,(\mu,z)\in join(X_1,...,X_k)\Big\}.\end{align*}
\end{deff}
The \jenvelope{} is the supremum of the join,  in much the same way as the concave envelope, $V^{co}(\mu|X_1,...,X_k)$, is the supremum of the convex hull. Though formally the \jenvelope{} and concave envelope are defined as the supremum, in the environment that we study in this paper, the supremum is attained in the set and can be replaced by maximum.\footnote{In other settings where the maximum is not attained in the set, our characterization can be applied to the join's closure, delivering a close approximation.} In Figure \ref{Fig:Adjud}, the purple curve represents the \jenvelope{} of the graphs of the sender's payoff function.

\begin{deff}(Interim Payoff Graphs). Sender's interim payoff graph for action $s$ is the following subset of $\Delta(\Omega)\times \mathbb{R}$, $\widehat{v}_s\equiv\{(\mu,\widehat{v}(\mu,s))\,|\,\mu\in\Delta(\Omega)\}.$

\end{deff}

Figure \ref{Fig:Adjud} hints at the connection between the sender's problem, the interim payoff graphs, their join, and the \jenvelope{}, which we make precise in the following proposition.

\begin{prop}\label{sigcom}(Signaling With Commitment). In sender's problem,
\begin{enumerate}
\item[(i)]  sender can attain payoff $z$ if and only if $(\mu_0,z)\in join(\widehat{v}_s)_{s\in S}$.

\item[(ii)]  for $(\mu_0,z)\in join(\widehat{v}_s)_{s\in S}$, some $\{\lambda_s\}_{s\in S}$ exist such that $\lambda_s\geq 0$, $\sum_{s\in S}\lambda_s=1$, and $(\mu_0,z)=\sum_{s\in S}\lambda_s(\mu_s,\widehat{v}(\mu_s,s))$. Sender attains $z$ with belief system $\{\mu_s,\tau(\mu_s)=\lambda_s\}_{s\in S}$.
	
\item[(iii)] sender's largest attainable payoff with prior $\mu_0$ is  $V^{jo}(\mu_0|(\widehat{v}_s)_{s\in S})$.	
 
\end{enumerate}
\end{prop}
To simplify notation, we abbreviate $V^{jo}(\mu_0|(\widehat{v}_s)_{s\in S})$ to $V^{jo}(\mu_0)$.

\paragraph{Robustness Interpretation of the Join.} In certain environments, the join has a noteworthy interpretation: it characterizes the set of Perfect Bayesian Equilibria that are robust to sender's unobserved commitment power. In particular, consider a signaling game in which sender may be an arbitrary commitment type. With probability $p$, the sender's action is the realization of an \textit{exogenous} commitment strategy, $\pi_C(\cdot|\omega)$. With complementary probability, the action is freely chosen by the sender after observing the state. Sender and receiver know the probability of the commitment type and the commitment strategy, but receiver cannot determine whether the sender's realized action was drawn from the commitment strategy or was sender's free choice. An analyst observes this signaling game and is aware that sender may be a commitment type. However, the analyst cannot directly observe the probability $p$, or the commitment strategy. Nevertheless, the analyst would like to determine the set of  sender payoffs (beliefs, and strategies) that can arise in Perfect Bayesian Equilibria across all possible values of these unobservables. With the advent of algorithms that offer product advice, trade in markets, or otherwise engage in signaling activity, such an exercise is particularly relevant.

Assuming the receiver's optimal action is unique at each $(\mu,s)$, Proposition \ref{sigcom} provides an answer. With this assumption, in any PBE sender's interim payoff is $\widehat{v}(\mu_s,s)$.\footnote{If receiver had multiple best responses, then in equilibrium, receiver might select an action that is not sender-preferred. Consequently, the sender's interim payoff may differ from $\widehat{v}(\mu_s,s)$.} Furthermore, a PBE of the signaling game with the commitment type induces a Bayes-Plausible belief system. By Proposition \ref{sigcom}, the equilibrium payoff must belong to the join. Conversely, any point in the join can be generated by a Bayes-Plausible belief system. If the sender is always a commitment type whose strategy induces this belief system, then the beliefs, receiver responses, and the commitment strategy form a PBE.

\subsection{Extended Commitment}\label{ACREC}
In the main model, the sender can transmit information to the receiver only through his action. In the motivating example, the organization is legally barred from communicating with stakeholders about the merits of specific grievances; only the decision to dismiss or address the grievance is observed. To understand the implications of such restrictions, in this section we study commitment power in the cheap talk extension of the signaling game. In particular, we augment the strategy set $S$ with a large message space $M$, and we allow sender to commit to joint distributions $\pi_E(\cdot,\cdot|\omega)$ over messages and actions $(m,s)\in M\times S$ conditional on state $\omega\in \Omega$. Thus, sender can transmit information using his action, message, or any combination of the two. The players' payoffs depend only on $(s,a,\omega)$---messages do not directly affect payoffs. We refer to this environment as ``extended commitment.'' 

 Though the extended commitment benchmark allows sender to transmit a single public message, we show that  allowing the sender to design the game's information structure along with his strategy does not increase his payoff further.
 \begin{prop}\label{infodesign}(Extended Commitment vs. Designing Information Structure). Suppose sender can design message spaces $M_\sigma,M_\rho$ and a joint distribution $\pi_{I}(\cdot,\cdot,\cdot|\omega)$ over private messages and a public action $(m_\sigma,m_\rho,s)\in M_\sigma\times M_\rho\times S$, conditional on state $\omega$. Any payoff that sender can achieve in such a setting can also be achieved in the extended commitment benchmark.
 \end{prop}
Thus, the extended commitment benchmark delivers the highest payoff that can be achieved in a signaling game without changing the payoff functions, the prior belief, or the extensive form.

We apply the belief-based approach to characterize the set of attainable payoffs. First, note that each public realization $(m,s)$ is associated with a posterior belief $\mu\in\Delta(\Omega)$. Thus, each sender strategy $\Pi_E$ induces a finite joint distribution $\tau$ of the posterior belief and sender action over $G\equiv \Delta(\Omega)\times S$, and all strategies that induce the same joint distribution $\tau$ are outcome-equivalent. In the usual way, this joint distribution can be decomposed into a marginal for belief $\mu$, denoted $\tau_m$, and a distribution for the action $s\in S$ conditional on $\mu$, denoted $\tau_c$, i.e., $\tau(\mu,s)=\tau_m(\mu)\tau_c(s|\mu)$. From the law of iterated expectations, the marginal distribution of the posterior belief must be Bayes-Plausible:  $\sum\tau_m(\mu)\mu=\mu_0$, where the summation is over  the support of $\tau_m$. Thus, Bayes-Plausibility of $\tau_m$ is a necessary condition for joint distribution $\tau$ to be induced by some sender strategy. Indeed, this is also sufficient, as we show in the following result.

\begin{prop}(Inducible Joint Distribution).\label{inducible} With extended commitment, 
\begin{enumerate}
\item[(i)]  joint distribution $\tau$ is induced by some sender strategy if and only if the marginal distribution of the posterior belief $\tau_m$ is Bayes-Plausible.
\item[(ii)] if joint distribution $\tau$ can be induced by some sender strategy, then it can also be induced by a strategy in which observing the sender's action $s$, does not convey any additional information, beyond what is learned from the sender's message $m$.
\end{enumerate}
\end{prop}

Intuitively, whenever the marginal distribution of the belief, $\tau_m$, is Bayes-Plausible, the joint distribution $\tau$ can be induced as the result of a two-step procedure. Design a sender strategy $\pi(\cdot|\omega)$ that induces $\tau_m$ using messages alone; this can always be done if $\tau_m$ is Bayes-Plausible.  In the first step, generate a realization from this strategy. If realized posterior belief $\mu$ is drawn in the first step, then select an action according to conditional distribution $\tau_c(\cdot|\mu)$ in the second step. By construction, the joint probability of a belief and action produced in this manner is $\tau_m(\mu)\tau_c(s|\mu)=\tau(\mu,s)$. Furthermore, the distribution of the action depends only on the realized posterior belief in the first stage, which itself depends only on the realized \textit{message}. Conditional on the message, the action conveys no information about the state. In this sense, extended commitment allows information provision to be uncoupled from the choice of action.

Building on Proposition \ref{inducible}, part (i), with extended commitment, sender's problem is to design a joint distribution of belief and actions $\tau$, that maximizes his expected payoff  subject to the constraint that the marginal distribution of the belief $\tau_m$ is Bayes-Plausible. 
\begin{align*}
    \max_{\tau\in\Delta(G)}\sum_{(\mu,s)\in\text{sp}[\tau]}\tau(\mu,s)\widehat{v}(\mu,s)\quad\text{subject to}\quad\sum_{\mu\in\text{sp}[\tau_m]}\tau_m(\mu)\mu=\mu_0,
\end{align*}
where $\text{sp}[\cdot]$ denotes the support of the distribution.

To solve the sender's problem, note first that any attainable payoff for the sender at prior $\mu_0$ can be written as a convex combination,
\begin{align*}
	(\mu_0,V(\mu_0))=\Big(\sum_{\mu\in\text{sp}[\tau_m]}\tau_m(\mu)\mu,\sum_{(\mu,s)\in \text{sp}[\tau]}\tau(\mu,s)\widehat{v}(\mu,s)\Big)=\sum_{\mu\in \text{sp}[\tau_m]}\tau_m(\mu)\Big(\mu,\sum_{s\in S}\tau_c(s|\mu)\widehat{v}(\mu,s)\Big).
\end{align*}
Next, note that $\sum_{s\in S}\tau_c(s|\mu)\widehat{v}(\mu,s)$ is  a convex combination of the sender's interim payoff functions evaluated at posterior belief $\mu$. Graphically, it is a convex combination of sender's payoff functions in the ``vertical dimension,'' where the belief is fixed and the action $s$ varies. From Proposition \ref{inducible}, sender can freely select the  conditional distribution of actions $\tau_c(\cdot|\mu)$ at each posterior belief.\footnote{The decomposition of the joint distribution $\tau$ into a conditional $\tau_c(\cdot|\mu)$ and a marginal $\tau_m$, and the expression for the sender's payoff in the text also apply in the main model. However, in the main model, sender faces additional constraints on the conditional distributions $\tau_c(\cdot|\mu)$. In particular, in signaling with commitment, each action $s$ in sender's chosen belief system generates a single belief, $\mu_s$. Thus, if $\tau_c(s|\mu)>0$ for some $\mu$, then $\tau_c(s|\mu')=0$, for all $\mu'\neq\mu$. } Thus, at a given posterior belief $\mu$, varying $\tau_c(\cdot|\mu)$ allows the sender to achieve any payoff between the highest and lowest payoff graphs at that belief. Simultaneously, by varying the Bayes-Plausible distribution of the posteriors $\tau_m$, sender can generate any convex combination of these ``vertical payoffs'' across posteriors (i.e., ``horizontally''). By varying the distribution of the posterior $\tau_m$ and the conditional distribution of actions $\tau_c(\cdot|\mu)$ together, sender can achieve any payoff that is in the convex hull of the union of the payoff graphs. Obviously, under the sender's optimal strategy, for each realization of the posterior belief, sender selects the action(s) that generate the highest payoff at that belief; we refer to such actions as ``belief-optimal.''
\begin{deff} (Belief-Optimal Actions). Sender action $s\in S$ is belief-optimal at $\mu$, if and only if $\widehat{v}(\mu,s)\geq \widehat{v}(\mu,s')$ for all $s'\in S$.
\end{deff}
In other words, a belief-optimal action at $\mu$ is the best available action for sender, if both he and receiver share posterior belief $\mu$---crucially, a belief-optimal action at $\mu$ may not be optimal if sender knows the realized state $\omega$.

We have thus proved the following proposition.  

\begin{prop}(Extended Commitment).\label{infact} With extended commitment, 
\begin{enumerate}
\item[(i)] sender can achieve any payoff inside the convex hull of the union of his payoff graphs, $con((\widehat{v}_s)_{s\in S})$, and his maximum payoff at prior $\mu$ is their concave envelope $V^{co}(\mu)$.

\item[(ii)] in sender's optimal strategy, if $\tau(\mu,s)>0$, then $s$ is belief-optimal at $\mu$.

\end{enumerate}	
\end{prop}

 \begin{figure}
\begin{minipage}{0.5\textwidth}
\begin{tikzpicture}[scale=0.7]

%AXES
\draw[line width=0.5pt] (0,-6.5) -- (0,4); 
\fill (8,-2) node[right] {\footnotesize{$\mu$}};
\draw[line width=0.5pt] (2,-2) -- (8,-2); 
% \draw[line width=0.5pt] (0,0) -- (8,0); 

%AXES LABELS
% \fill (8,0) node[right] {\footnotesize{$\mu$}};

%$\widehat{v}(\mu,a)$

\draw[line width=1pt,blue] (0,0)--(4,0);
\draw[line width=1pt,blue] (4,4)--(8,4);
\filldraw[color=blue, fill=white] (4,4) circle (3pt);
\draw[line width=0.5pt,blue,dashed] (4,0)--(4,4);
\fill (6,4.5) node[right] {\footnotesize{$\widehat{v}(\cdot,d)$}};
\fill (4,-2) node[below] {\footnotesize{$\theta_d$}};
\draw[line width=0.5pt,black,dotted] (4,-2)--(4,0);
\fill (0,-2) node[left] {\footnotesize{$\mu_a^*$}};
\fill (4,4) node[above] {\footnotesize{$\mu_d^*$}};

\fill (0,0) node[left] {\footnotesize{$1-l$}};
\draw[line width=0.5pt,black,dotted] (4,4)--(0,4);
\fill (0,4) node[left] {\footnotesize{$1$}};
\fill (0,-6) node[left] {\footnotesize{$-l$}};
\draw[line width=0.5pt,black,dotted] (0,-6)--(2,-6);
\fill (2.5,-2) node[below] {\footnotesize{$\theta_a$}};

\draw[line width=1pt,violet] (4,4)--(8/6,0);
\draw[line width=1pt,violet] (4,4)--(8,4);
\draw[line width=1pt,violet] (0,0)--(8/6,0);

% \fill (8/6+0.15,0) node[below] {\footnotesize{$\widehat\mu$}};

%$\widehat{v}(\mu,m)$
\draw[line width=0.5pt,black,dashed] (0,-2)--(8/6,0);
\draw[line width=1pt,red] (0,-2)--(2,-2);
\draw[line width=1pt,red] (2,-6)--(8,-6);
\draw[line width=0.5pt,red,dashed] (2,-6)--(2,-2);
 \fill (6,-5.5) node[right] {\footnotesize{$\widehat{v}(\cdot,a)$}};

 \draw[line width=0.5pt,black,dashed] (2,-6)--(0,0);

% \draw[line width=0.5pt,black,dashed] (0,-2)--(8/6,0);
%Join 2
%\draw[line width=0.5pt,black,dashed] (8/3,2)--(0.73,0.53);
%\draw[line width=0.5pt,black,dashed] (0.71,0.15)--(8/3,2);
%ENVELOPE
%\draw[line width=1pt,violet] (0,0)--(0.73,0.53);

% \draw[line width=1pt,violet] (4,4)--(8/6,0);
% \draw[line width=1pt,violet] (4,4)--(8,4);
% \draw[line width=1pt,violet] (0,0)--(8/6,0);
% \draw[line width=1pt,green] (0,0)--(4,4);
%\draw[line width=1pt,violet] (4,4.07)--(8,4.07);
%\draw[line width=0.5pt,black, dashed] (0.73,0.53)--(0.73,0);
\filldraw[color=red, fill=white] (0,-2) circle (3pt);
\filldraw[color=blue, fill=white] (4,4) circle (3pt);
\filldraw[color=blue, fill=white] (0,0) circle (3pt);

\filldraw[color=red, fill=white] (2,-6) circle (3pt);

\fill[violet,opacity=0.05] (0,-2)--(0,0)--(8/6,0)--(4,4)--(8,4)--(8,-6)--(2,-6)--(0.7,-2)--(0,-2);
\fill (8/6,-2) node[below] {\footnotesize{$\widehat\mu$}};
 \draw[line width=0.5pt,black,dotted] (8/6,-2)--(8/6,0);
\end{tikzpicture}
\end{minipage}
 \begin{minipage}{0.5\textwidth}
\begin{tikzpicture}[scale=0.7]

%AXES
\draw[line width=0.5pt] (0,-6.5) -- (0,4); 
\fill (8,-2) node[right] {\footnotesize{$\mu$}};
\draw[line width=0.5pt] (2,-2) -- (8,-2); 
% \draw[line width=0.5pt] (0,0) -- (8,0); 

%AXES LABELS
% \fill (8,0) node[right] {\footnotesize{$\mu$}};

%$\widehat{v}(\mu,a)$

\draw[line width=1pt,blue] (0,0)--(4,0);
\draw[line width=1pt,blue] (4,4)--(8,4);
\filldraw[color=blue, fill=white] (4,4) circle (3pt);
\draw[line width=0.5pt,blue,dashed] (4,0)--(4,4);
\fill (6,4.5) node[right] {\footnotesize{$\widehat{v}(\cdot,d)$}};
\fill (4,-2) node[below] {\footnotesize{$\theta_d$}};
\draw[line width=0.5pt,black,dotted] (4,-2)--(4,0);
\fill (0,0) node[left] {\footnotesize{$\mu_B$}};
\fill (4,4) node[above] {\footnotesize{$\mu_G$}};
\fill (0,-2) node[left] {\footnotesize{$0$}};

% \fill (0,0) node[left] {\footnotesize{$1-l$}};
\draw[line width=0.5pt,black,dotted] (4,4)--(0,4);
\fill (0,4) node[left] {\footnotesize{$1$}};
\fill (0,-6) node[left] {\footnotesize{$-l$}};
\draw[line width=0.5pt,black,dotted] (0,-6)--(2,-6);
\fill (2.5,-2) node[below] {\footnotesize{$\theta_a$}};

\draw[line width=1pt,green] (4,4)--(8,4);

% \fill (8/6+0.15,0) node[below] {\footnotesize{$\widehat\mu$}};

%$\widehat{v}(\mu,m)$
% \draw[line width=0.5pt,black,dashed] (0,-2)--(8/6,0);
\draw[line width=1pt,red] (0,-2)--(2,-2);
\draw[line width=1pt,red] (2,-6)--(8,-6);
\draw[line width=0.5pt,red,dashed] (2,-6)--(2,-2);
 \fill (6,-5.5) node[right] {\footnotesize{$\widehat{v}(\cdot,a)$}};

 \draw[line width=0.5pt,black,dashed] (2,-6)--(0,-2);
 
\draw[line width=1pt,green] (0,0)--(4,4);

% \draw[line width=0.5pt,black,dashed] (0,-2)--(8/6,0);
%Join 2
%\draw[line width=0.5pt,black,dashed] (8/3,2)--(0.73,0.53);
%\draw[line width=0.5pt,black,dashed] (0.71,0.15)--(8/3,2);
%ENVELOPE
%\draw[line width=1pt,violet] (0,0)--(0.73,0.53);

% \draw[line width=1pt,violet] (4,4)--(8/6,0);
% \draw[line width=1pt,violet] (4,4)--(8,4);
% \draw[line width=1pt,violet] (0,0)--(8/6,0);
% \draw[line width=1pt,green] (0,0)--(4,4);
%\draw[line width=1pt,violet] (4,4.07)--(8,4.07);
%\draw[line width=0.5pt,black, dashed] (0.73,0.53)--(0.73,0);
\filldraw[color=blue, fill=white] (4,4) circle (3pt);
\filldraw[color=blue, fill=white] (0,0) circle (3pt);

\filldraw[color=red, fill=white] (0,-2) circle (3pt);
\filldraw[color=red, fill=white] (2,-6) circle (3pt);

\fill[green,opacity=0.05] (0,-2)--(0,0)--(4,4)--(8,4)--(8,-6)--(2,-6)--(0,-2);

\end{tikzpicture}
\end{minipage}
\caption{\label{motex2} Motivating Example Continued.}
\end{figure}

Figure \ref{motex2} illustrates. The left panel depicts the set of attainable payoffs in the motivating example with commitment to actions only. The right panel illustrates the set of attainable payoffs with extended commitment. From the preceding discussion, we may imagine that the sender designs a Bayes-Plausible distribution of posteriors, $\tau_m$. For each realization of the posterior belief $\mu$ in the support, sender then selects the graph $\widehat{v}_s$ at which the payoff is evaluated, thereby assigning an action $s$ to realization $\mu$. By varying the distribution of posterior beliefs and the actions associated with each realization, the sender can achieve any payoff in the convex hull of the union of the sender's payoff graphs, shaded in green. The green curve is the concave envelope of the payoff graphs, evidently higher than the join envelope. While confidentiality restrictions that prevent organizations from communicating information about the details of individual grievances can be justified on the grounds of privacy protections, Figure \ref{motex2} highlights an additional rationale for such restrictions. As is evident in the right panel, with extended commitment the organization always selects belief-optimal action $d$. In contrast, with confidentiality restrictions,  the organization addresses valid grievances with positive probability at moderate prior beliefs. Thus, removing the restriction may reduce the probability that valid grievances are addressed. This is particularly concerning if the organization does not fully account for the benefits of addressing valid grievances.

Together Propositions \ref{infodesign} and \ref{infact} imply that in a signaling game, the highest payoff that sender can attain for fixed payoff functions and prior can be achieved using a combination of two familiar instruments, \textit{persuasion} and \textit{delegation}. Imagine that in a standard signaling game, sender has the capability to persuade: before learning the state, sender can design a statistical experiment that generates public information about the state. Given the optimal joint distribution in the extended commitment benchmark, $\tau$, sender can induce the marginal distribution of the posterior $\tau_m$, using the statistical experiment alone. Simultaneously, sender delegates authority over the action to an agent, whose payoffs are identical to his own. Crucially, the agent's only source of information about the state is the statistical experiment designed by the sender, which ensures that the agent prefers the belief-optimal action at each realized posterior.\footnote{In order to attain the extended commitment payoff, the action at each realized posterior $\mu$ must be belief-optimal (Proposition \ref{infact}, (ii)). In other words, it must be optimal given only \textit{public information}, summarized by the realized belief $\mu$. If the agent observes additional information about the state before deciding, then he may not wish to select the belief-optimal action following bad news; see Section \ref{sec:exben}.} By generating the public belief and the action in this manner, sender can recreate the optimal joint distribution from the benchmark with extended commitment. 
% {\color{blue}We provide a }

\subsection{Pre-Persuasion}
How important is the sender's ability to commit to actions in the extended commitment benchmark? To address this question, we introduce a benchmark in which the sender can commit to reveal information but cannot commit to future actions. As in the extended commitment benchmark, we augment the signaling game by introducing a large space $M$ of payoff-irrelevant messages. Before learning the state, sender commits to a statistical experiment, which determines the probability distribution of messages conditional on the state. After the experiment is publicly realized, sender privately observes the state and selects an action. Receiver observes the statistical experiment, its realization, and sender's action, and then selects a response. This benchmark is a standard signaling game augmented with an initial (Bayesian) persuasion stage---we therefore refer to this benchmark as ``pre-persuasion.'' Though we motivate this benchmark as a point of comparison to extended commitment, it may also be of independent interest.

To construct the set of attainable payoffs in the pre-persuasion benchmark, we first construct the \textit{equilibrium payoff graph.} At each belief $\mu$, a (possibly empty) set of Perfect Bayesian Equilibria exists in the signaling game. To make the pre-persuasion benchmark as powerful as possible, we do not impose any refinements and assume sender-preferred equilibrium selection.\footnote{ When PBE refinements are imposed or equilibrium selection is not favorable to the sender, our characterization provides an upper bound for the sender's payoffs in  pre-persuasion environments.} Each equilibrium can be represented as a Bayes-Plausible belief system---each sender action $s\in S$ induces a posterior belief for the receiver (possibly an off-path belief) $\mu_s$, and these beliefs must average to the prior. Sender's expected equilibrium payoff is therefore a convex combination of the interim payoff graphs, which lies weakly below the \jenvelope{} at $\mu$. We define $\widehat{v}_{eq}(\mu)$ to be the largest payoff that sender can attain in a PBE of the signaling game when the prior belief is $\mu$, and $\widehat{v}_{eq}$ to be its graph in $\Delta(\Omega)\times \mathbb{R}$. If no PBE exists at $\mu$, then we set $\widehat{v}_{eq}(\mu)=-\infty$.

Applying the logic of Bayesian persuasion, sender can attain any payoff in the convex hull of $\widehat{v}_{eq}$ by generating a Bayes-Plausible distribution of posterior beliefs in the persuasion stage.
Let $V^{pp}(\cdot)$ denote the concave envelope of the equilibrium payoff graph
\begin{align*}
V^{pp}(\mu)\equiv  \sup\Big\{z\,|\,(\mu,z)\in con(\widehat{v}_{eq})\Big\}.
\end{align*}
To simplify the exposition, we focus on settings in which the supremum is attained within the set (e.g., $|A|<\infty$). In such environments,  $V^{pp}(\cdot)$ is sender's highest attainable payoff in the pre-persuasion benchmark.

In the proposition that follows, pooling equilibria play an important role. For clarity, we define pooling PBE explicitly, allowing for prior beliefs that do not have full support.

\begin{deff} A pooling equilibrium exists at $(\mu,s)\in \Delta(\Omega)\times S$ if and only if a belief system $\mathcal{B}$ exists such that (1) $\mu_s=\mu$, $\tau(\mu_s)=1$, and (2) for all $\omega$ such that $\mu(\omega)>0$ and $s'\in S$,
\begin{align*}
    v(s,\widehat{a}(\mu,s),\omega)\geq v(s',\widehat{a}(\mu_{s'},s'),\omega).
\end{align*}
\end{deff}
Condition (1) imposes consistency of beliefs---it requires that when the on-path action $s$ is observed beliefs do not update from $\mu$, without imposing any restriction on the rest of the belief system (the off-path beliefs $\mu_{s'}, s'\neq s$).\footnote{We require receiver's response to be optimal at off-path beliefs. That is, receiver's response is sequentially rational at each information set given the receiver's (possibly off-path) beliefs and the sender's action. } Condition (2) is sender's incentive constraint: it requires that for all states in the support of belief $\mu$, action $s$ is optimal for sender, accounting for the receiver's responses. 

To provide a simpler characterization of $V^{pp}(\cdot)$, we define two related objects. Let the \textit{pooling set}, $P$ be the following subset of $(\mu,s)\in \Delta(\Omega)\times S$,
\begin{align*}
P\equiv\{(\mu,s)\,|\, (\mu,s)\in \Delta(\Omega)\times S\, \cap\, \text{a pooling equilibrium exists at }(\mu,s)\}.
\end{align*}
That is, the pooling set contains all pairs of beliefs and actions at which a pooling equilibrium exists. Building on this, we define the \textit{pooling payoff graphs}, $\widehat{v}_{s}^{P}$.
\begin{deff}(Pooling Payoff Graphs). The pooling payoff graph for action $s$ is the following subset of $\Delta(\Omega)\times \mathbb{R}$: $\widehat{v}^{P}_s\equiv\{(\mu,\widehat{v}(\mu,s))\,|\,\mu\in\Delta(\Omega)\text{ and }(\mu,s)\in P\}.$
\end{deff}
These graphs represent the sender's equilibrium payoffs when focusing only on the game's pooling equilibria. Though pooling equilibria are only a subset of all equilibria, in the following proposition we show that considering only pooling equilibria is sufficient to determine the pre-persuasion payoff.

\begin{prop}(Pre-Persuasion).\label{pre-per} In the pre-persuasion benchmark, sender can achieve any payoff in the convex hull of the union of the pooling payoff graphs, $con((\widehat{v}^P_s)_{s\in S})$, and his maximum payoff at prior $\mu$ is their concave envelope,
\begin{align*}
    V^{pp}(\mu)=\sup\Big\{z\,|\,(\mu,z)\in con((\widehat{v}^P_{s})_{s\in S})\Big\}.
\end{align*}
\end{prop}
The key step of the proof shows that $\widehat{v}_{eq}\subseteq con(\cup_{s\in S}\widehat{v}^P_s)$. That is, any point on the equilibrium payoff graph can be represented as a convex combination of points arising from \textit{pooling} equilibria.\footnote{The converse is not true---a convex combination of points on the pooling payoff graphs might not be on the equilibrium payoff graph.} In other words, the extreme points of the equilibrium payoff graph must arise from pooling equilibria. It follows that any convex combination of points on the equilibrium payoff graph (i.e., any point in $con(\widehat{v}_{eq})$) is, itself, a convex combination of  points arising from pooling equilibria (i.e., it is also in $con(\cup_{s\in S}\widehat{v}_s^P)$). 

This characterization is significant for several reasons. At a practical level, it allows the analyst to focus only on pooling equilibria, which simplifies study of the pre-persuasion benchmark. Conceptually, this characterization also gives insight into the tradeoffs between signaling with commitment (where the sender can commit to actions but cannot send messages), pre-persuasion (where the sender can commit to send messages but cannot commit to actions), and extended commitment (where the sender can commit to both). Indeed, in signaling with commitment, sender can connect any point on one interim payoff graph to a point on a different graph, but he cannot connect two points on the same graph. With pre-persuasion, sender can only connect points on the pooling payoff graphs (a subset of the interim payoff graphs), but he can connect any such point to any other. With extended commitment, sender is not restricted in either dimension: he is free to connect any point on the interim payoff graphs to any other. We elaborate further on these comparisons in the next section.

\subsection{Is Extended Commitment Beneficial?}\label{sec:exben}
To study how messages and actions interplay in more detail, we compare extended commitment with the more restrictive settings we have analyzed. First, we characterize environments in which signaling with commitment and extended commitment deliver the same payoff: in such environment's the sender's ability to provide additional information beyond what is inferred from the realized action doesn't benefit the sender. Second, we compare signaling with commitment with the pre-persuasion benchmark, in which sender can design a statistical experiment that reveals public information ahead of a standard signaling game. We characterize environments in which the extended commitment benchmark does strictly better; in such environments, the ability to commit to actions is valuable to sender even if he has the ability to persuade.

\paragraph{Extended Commitment vs. Signaling with Commitment.} Although sender generally exploits both commitment to actions and the communication protocol in the extended commitment benchmark, in certain environments, sender can achieve the same payoff with commitment to actions only. In other words, given the prior belief and the players' payoff functions, signaling with commitment is enough for sender to achieve the extended commitment payoff. Intuitively, in signaling with commitment, the realized action plays two roles---it transmits information about the state, and it directly affects payoffs. As we have seen, these two roles are generally at odds, and the sender can benefit by uncoupling them. However, in a class of environments which we characterize below, both of these roles can be fulfilled simultaneously by realized actions without any adverse consequences for the sender.

 To see how this works, it is helpful to think of sender choosing his strategy in the following way. Imagine that sender uses a public message to generate a Bayes-Plausible posterior belief distribution. Once the belief is realized, sender selects his optimal strategy in the ``signaling with commitment'' game, with the realized belief $\mu$ playing the role of the prior. Consistent with Proposition \ref{sigcom}, sender can achieve any payoff in the join of the payoff graphs, and his highest attainable payoff is their \jenvelope{} $V^{jo}(\mu)$. Thus, when sender is designing the posterior belief distribution initially, he anticipates payoff function $V^{jo}(\cdot)$. Therefore, by designing a message protocol in the first stage, sender can achieve any payoff in the convex hull of the join, and his highest attainable payoff is the concave envelope of the join itself. In other words, with extended commitment, sender's optimal payoff can be found by concavifying the join of the sender's payoff graphs. We show this formally in Lemma \ref{infact4} in the Appendix. Part (i) of the following proposition follows from familiar arguments.

\begin{prop}\label{infact2}(Extended Commitment vs. Signaling With Commitment). 
\begin{enumerate}
\item[(i)]
At all prior beliefs, sender's equilibrium payoff in signaling with commitment is identical to his payoff with extended commitment, if and only if the \jenvelope{} of sender's payoff graphs is concave on $\Delta(\Omega)$. 
\item[(ii)] If the sender's interim payoff function $\widehat{v}(\cdot,s)$ is concave on $\Delta(\Omega)$ for all $s\in S$, then the \jenvelope{} is concave  on $\Delta(\Omega)$. 
\end{enumerate}

\end{prop}

Part (i) of the Proposition provides a necessary and sufficient condition for the payoffs under extended commitment and signaling with commitment to coincide at \textit{all} prior beliefs. In other words, if this condition holds, then a sender who can commit to actions (signaling with commitment) does not benefit from the ability to transmit additional messages (extended commitment), regardless of the prior information about the state. Part (ii) of the proposition provides a straightforward sufficient condition. Although (ii) requires that each interim payoff graph is concave, it does not rule out information provision by the sender,  as illustrated in the application to nominal ratings (Proposition \ref{spenceRA}). 

 \paragraph{Extended Commitment vs. Pre-Persuasion.} Building on our characterizations of the extended commitment and pre-persuasion benchmarks, we can identify settings in which the extended commitment payoff $V^{co}(\cdot)$ is strictly larger than the pre-persuasion payoff $V^{pp}(\cdot)$. In such settings, sender's ability to commit to actions is  essential for him to attain the extended commitment payoff. 

\begin{prop}\label{prop:preper}(Extended Commitment vs. Pre-persuasion).  Let $T^{co}\subseteq\Delta(G)$ be the set of joint distributions of beliefs and actions that is optimal in the extended commitment benchmark for prior belief $\mu_0$. The extended commitment payoff is achieved in the pre-persuasion benchmark for prior belief $\mu_0$, if and only if there exists a $\tau^{co}\in T^{co}$ with the following property: the support of $\tau^{co}$ is a subset of the pooling set $P$, i.e., for each $(\mu,s)$ such that $\tau^{co}(\mu,s)>0$, a pooling equilibrium exists at $(\mu,s)$.

\end{prop}

Though it is a powerful instrument, in general, persuasion alone does not allow sender to attain the extended commitment payoff, $V^{co}(\cdot)$. To see the idea, suppose that the optimal joint distribution with extended commitment is unique, and let  $\tau^{co}_m$ be the associated marginal distribution of the belief. Imagine that in the persuasion stage, sender designs a statistical experiment that induces $\tau^{co}_m$. To achieve the extended commitment payoff, the sender's action at each realized posterior $\mu$ must be belief-optimal (Proposition \ref{infact}, (ii)). In other words, sender must select an action that is optimal given only \textit{public information}, summarized by the realized belief $\mu$. However, in the pre-persuasion benchmark, sender learns the state privately before selecting an action, and subsequent play must be an equilibrium of the signaling game. If a pooling equilibrium on the belief-optimal action exists at each $\mu\in\text{sp}[\tau_m^{co}]$, then sender can recreate the extended commitment payoff with this simple persuasion strategy. However, if such a pooling equilibrium fails to exist for some $\mu\in\text{sp}[\tau_m^{co}]$, then this point is not available to the sender in the pre-persuasion benchmark, and he cannot attain the extended commitment payoff.

\subsection{Example: Beer-Quiche Game}
Our findings can be illustrated in a version of the beer-quiche game. The tough sender's payoff is 1 if he chooses $B$ (beer) and 0 if he chooses $Q$ (quiche). The wimpy sender's gains 1 from selecting quiche and 0 from selecting beer. Furthermore, wimpy sender's payoff is reduced by $c>1$ if the receiver selects $F$ (fight). Receiver's payoff is 0 if she chooses $A$ (leave alone), is $k\in(0,1/2)$ if she chooses $F$ and sender is wimpy, and is $-(1-k)$ if she chooses $F$ and sender is tough. Let $\mu$ be the belief that the sender is tough. 

To construct the interim payoff graphs, note that the receiver's expected payoff of $F$ at belief $\mu$ is  $(1-\mu)k-\mu (1-k)=k-\mu$, and thus receiver's optimal response is $\widehat{a}(\mu_s,s)=F$ if and only if $\mu<k$.\footnote{When indifferent, we assume that receiver chooses $A$, which is sender-preferred.}  Sender's interim payoffs are therefore,
\begin{align*}
    \widehat{v}(\mu,B)=\mu-c(1-\mu)\mathcal{I}(\mu<k)\quad\quad
    \widehat{v}(\mu,Q)=(1-\mu)(1-c\mathcal{I}(\mu<k)).
\end{align*}

Without commitment, the tough type of sender always selects beer in equilibrium. Thus, an equilibrium in which sender pools on quiche exists if and only if sender is known to be wimpy, $\mu=0$. It is also immediate that a separating equilibrium does not exist in this game for any prior belief.\footnote{If the receiver expects sender to select quiche if and only if he is wimpy, then he bullies when sender selects quiche and leaves sender alone when sender selects beer. By selecting quiche, wimpy sender obtains payoff $1-c<0$, but he could obtain $0$ by deviating to beer.} Thus, without commitment only two types of equilibrium are possible when $\mu>0$: pooling on beer, and a semi-separating equilibrium in which the tough type selects beer and the wimpy type mixes. It is straightforward to verify that the pooling equilibrium exists for $\mu\geq k$ and the semi-separating equilibrium exists for $0<\mu\leq k$. Noting that the payoff in the pooling equilibrium is $\mu$ and the payoff in the semi-separating equilibrium is $\mu+(1-\mu)(1-c)$, we have 
\begin{align*}
    \widehat{v}_{eq}(\mu)=\mu+\mathcal{I}(\mu< k)(1-\mu)(1-c).
\end{align*}
The graphs $\widehat{v}_Q$ (red), $\widehat{v}_B$ (blue), and $\widehat{v}_{eq}$ (brown) are illustrated in the left panel of Figure \ref{BeerQuiche}. Note that $\widehat{v}_{eq}$ (brown) coincides with $\widehat{v}_B$ (blue) above $k$. In the right panel $V^{jo}(\cdot)$ is drawn in purple, while $V^{co}(\cdot)$ is drawn in green for $\mu<k$ (for $\mu\geq k$ these two coincide).  

\begin{figure}
 \begin{minipage}{0.5\textwidth}
\begin{tikzpicture}[scale=0.45]

\pgfmathsetmacro{\c}{1.1}
\pgfmathsetmacro{\k}{(0.4)}
\pgfmathsetmacro{\kbl}{\k-\c*(1-\k)}
\pgfmathsetmacro{\kql}{(1-\k)*(1-\c)}
 \pgfmathsetmacro{\oql}{(1-\c)}
 \pgfmathsetmacro{\kqr}{1-\k}
 \pgfmathsetmacro{\keq}{1-\c+(\c)*\k}
% \pgfmathsetmacro{\kbr}{\k-\c*(1-\k)}

\draw[line width=1pt,blue] (0,-10*\c)--(10*\k,10*\kbl);
\draw[line width=1pt,blue] (10*\k,10*\k)--(10,10);

\draw[line width=1pt,red] (0,10*\oql)--(10*\k,10*\kql);
\draw[line width=1pt,red] (10*\k,10*\kqr)--(10,0);
\filldraw[color=red, fill=red] (4,6) circle (3pt);

 \draw[line width=1pt,brown,solid] (0,10*\oql)--(10*\k,10*\keq);
 \fill (10*\k,10*\keq-1) node[right] {\footnotesize{$\widehat{v}_{eq}(\cdot)$}};

%AXES
\draw[line width=0.5pt] (0,-11) -- (0,10); 
\draw[line width=0.5pt] (0,0) -- (10,0); 

% \draw[line width=1pt,blue] (0,-\c)--(10,10);

%AXES LABELS
\fill (10,0) node[right] {\footnotesize{$\mu$}};

%$\widehat{v}(\mu,a)$
 % \draw[line width=1pt,blue] (4,4)--(10,10);
 % \draw[line width=1pt,blue] (0,-8)--(4,-0.8);
 \draw[line width=0.5pt,black,dotted] (10*\k,10*\kbl)--(10*\k,10*\kqr);
 \draw[line width=0.5pt,black,dotted] (0,4)--(4,4);
  \draw[line width=0.5pt,black,dotted] (4,6)--(4,4);
  \draw[line width=0.5pt,black,dotted] (0,10)--(10,10);
  \draw[line width=0.5pt,black,dotted] (4,6)--(0,6);
% \filldraw[color=blue, fill=white] (4,7) circle (3pt);
% \draw[line width=0.5pt,blue] (4,0)--(4,7);
 \fill (10,10) node[right] {\footnotesize{$\widehat{v}(\cdot,B)$}};
\fill (4.5,0) node[below] {\footnotesize{$k$}};
\fill (0,4) node[left] {\footnotesize{$k$}};
\fill (0,6) node[left] {\footnotesize{$1-k$}};
\fill (0,-10*\c) node[left] {\footnotesize{$-c$}};
\fill (0,10*\oql) node[left] {\footnotesize{$1-c$}};
\fill (0,10) node[left] {\footnotesize{$1$}};

%$\widehat{v}(\mu,m)$
% \draw[line width=1pt,red] (0,2)--(4,1.2);
% \draw[line width=1pt,red] (4,6)--(10,0);
% \filldraw[color=red, fill=red] (4,6) circle (3pt);
% \draw[line width=0.5pt,red] (8/3,-1)--(8/3,2);
\fill (10,-1) node[right] {\footnotesize{$\widehat{v}(\cdot,Q)$}};
\filldraw[color=blue, fill=blue] (4,4) circle (3pt);
% \filldraw[color=blue, fill=white] (0,-8) circle (3pt);
% \filldraw[color=red, fill=red] (0,2) circle (3pt);

% \draw[line width=1pt,violet] (0,2)--(3.7,5)--(4,6)--(10,10);
% \draw[line width=0.5pt,black,dashed] (0,2)--(10,10);
% \draw[line width=0.5pt,black,dashed] (0,-8)--(4,6);
% \draw[line width=0.5pt,black,dashed] (10,10)--(4,6);

\end{tikzpicture}
\end{minipage}
 \begin{minipage}{0.5\textwidth}
\begin{tikzpicture}[scale=0.45]

\pgfmathsetmacro{\c}{1.15}
\pgfmathsetmacro{\k}{(0.4)}
\pgfmathsetmacro{\kbl}{\k-\c*(1-\k)}
\pgfmathsetmacro{\kql}{(1-\k)*(1-\c)}
 \pgfmathsetmacro{\oql}{(1-\c)}
 \pgfmathsetmacro{\kqr}{1-\k}
 \pgfmathsetmacro{\keq}{1-\c+(\c)*\k}
 \pgfmathsetmacro{\j}{\k/(2*(1-\k))}
 \pgfmathsetmacro{\jh}{-\c+\j*(\c+1-\k)/(\k)}
 
% \pgfmathsetmacro{\kbr}{\k-\c*(1-\k)}

\draw[line width=1pt,violet,solid] (10*\k,10*\kqr)--(10,10);

\draw[line width=0.5pt,blue] (0,-10*\c)--(10*\k,10*\kbl);
\draw[line width=.5pt,blue] (10*\k,10*\k)--(10,10);

\draw[line width=.5pt,red] (0,10*\oql)--(10*\k,10*\kql);
\draw[line width=.5pt,red] (10*\k,10*\kqr)--(10,0);
\filldraw[color=red, fill=red] (4,6) circle (3pt);

% equilibrium graph
\draw[line width=.5pt,brown,solid] (0,10*\oql)--(10*\k,10*\keq);
\fill (10*\k,10*\keq-1) node[right] {\footnotesize{$\widehat{v}_{eq}(\cdot)$}};

\draw[line width=1pt,violet,solid] (10*\k,10*\kqr)--(10,10);
\draw[line width=1pt,violet,solid] (10*\k,6)--(10*\j,10*\jh);
\draw[line width=1pt,violet,solid] (0,10*\oql)--(10*\j,10*\jh);
\draw[line width=1pt,green,solid] (0,10*\oql)--(10*\k,10-10*\k);

\draw[line width=0.5pt,black,dashed] (0,-10*\c)--(10*\j,10*\jh);
\draw[line width=0.5pt,black,dashed] (10*\k,10*\k)--(10*\j,10*\jh);

 \draw[line width=0.5pt,black,dotted] (10*\j,-0.4)--(10*\j,10*\jh);
 \fill (10*\j,-0.5) node[below] {\footnotesize{$\widetilde{\mu}$}};

%AXES
\draw[line width=0.5pt] (0,-10*\c) -- (0,10); 
\draw[line width=0.5pt] (0,0) -- (10,0); 

% \draw[line width=1pt,blue] (0,-\c)--(10,10);

%AXES LABELS
\fill (10,0) node[right] {\footnotesize{$\mu$}};

%$\widehat{v}(\mu,a)$
 % \draw[line width=1pt,blue] (4,4)--(10,10);
 % \draw[line width=1pt,blue] (0,-8)--(4,-0.8);
 \draw[line width=0.5pt,black,dotted] (10*\k,10*\kbl)--(10*\k,10*\kqr);
 \draw[line width=0.5pt,black,dotted] (0,4)--(4,4);
  \draw[line width=0.5pt,black,dotted] (4,6)--(4,4);
  \draw[line width=0.5pt,black,dotted] (0,10)--(10,10);
  \draw[line width=0.5pt,black,dotted] (4,6)--(0,6);
% \filldraw[color=blue, fill=white] (4,7) circle (3pt);
% \draw[line width=0.5pt,blue] (4,0)--(4,7);
 \fill (10,10) node[right] {\footnotesize{$\widehat{v}(\cdot,B)$}};
\fill (4.5,0) node[below] {\footnotesize{$k$}};
\fill (0,4) node[left] {\footnotesize{$k$}};
\fill (0,6) node[left] {\footnotesize{$1-k$}};
\fill (0,-10*\c) node[left] {\footnotesize{$-c$}};
\fill (0,10*\oql) node[left] {\footnotesize{$1-c$}};
\fill (0,10) node[left] {\footnotesize{$1$}};

%$\widehat{v}(\mu,m)$
% \draw[line width=1pt,red] (0,2)--(4,1.2);
% \draw[line width=1pt,red] (4,6)--(10,0);
% \filldraw[color=red, fill=red] (4,6) circle (3pt);
% \draw[line width=0.5pt,red] (8/3,-1)--(8/3,2);
\fill (10,-1) node[right] {\footnotesize{$\widehat{v}(\cdot,Q)$}};
\filldraw[color=blue, fill=blue] (4,4) circle (3pt);
\filldraw[color=red, fill=red] (4,6) circle (3pt);
% \filldraw[color=blue, fill=white] (0,-8) circle (3pt);
% \filldraw[color=red, fill=red] (0,2) circle (3pt);

% \draw[line width=1pt,violet] (0,2)--(3.7,5)--(4,6)--(10,10);
% \draw[line width=0.5pt,black,dashed] (0,2)--(10,10);
% \draw[line width=0.5pt,black,dashed] (0,-8)--(4,6);
% \draw[line width=0.5pt,black,dashed] (10,10)--(4,6);

\end{tikzpicture}\end{minipage}
\caption{\label{BeerQuiche} Beer-Quiche Game.}
\end{figure}

 We make a few observations based on inspection of Figure \ref{BeerQuiche}. (1) $V^{jo}(\cdot)$ lies strictly above the equilibrium payoff graph $\widehat{v}_{eq}$. By implication, the ability to commit to a strategy is valuable for the sender. (2) With commitment, the sender's strategy exhibits ``strategic inversion.'' In particular, in the middle interval $\mu\in[\widetilde{\mu},k]$, the optimal belief system has $\{\mu_B=0,\mu_Q=k\}$. By implication, the tough sender commits to always select $Q$, while the wimpy sender mixes between $Q$ and $B$. Thus, with positive probability, both types select the action that the other type prefers in the absence of signaling concerns.\footnote{Without commitment, the wimpy type mixes, but the tough type always selects $B$. Thus, the equilibrium without commitment features exaggeration or bluffing by the wimpy type. Because the tough always selects beer, however, a strategic inversion does not arise.} (3) The \jenvelope{} is convex on $[0,k]$. Consistent with Proposition \ref{infact2}, extended commitment is valuable for the sender. (4) The extended commitment payoff cannot be achieved without commitment to actions (i.e., in the pre-persuasion benchmark). Indeed, for  prior beliefs in $[0,k]$, the unique optimal joint distribution with extended commitment is supported on $\{(\mu=0,Q); (\mu=k,Q)\}$. Consistent with Proposition \ref{prop:preper}, to achieve the extended commitment payoff in the pre-persuasion benchmark, a pooling equilibrium on action $Q$ must exist at $\mu=k$. Such an equilibrium does not exist---the probability of the tough type is non-zero, and without commitment, the tough type always selects $B$. (5) Building on the previous point, the right panel of Figure \ref{BeerQuiche} shows that $V^{pp}(\cdot)<V^{co}(\cdot)$. Recalling that the brown line represents $\widehat{v}_{eq}$ for $\mu\leq k$ and the blue line represents it for $\mu\geq k$, the concave envelope $V^{pp}(\cdot)$ interpolates $(0,1-c)$ and $(k,k)$ for $\mu\leq k$, and follows the blue line for $\mu\geq k$. Thus, the pre-persuasion payoff is always strictly smaller than the extended commitment payoff. (6) Noting that  the pooling payoff graph $\widehat{v}_Q^P$ consists of a single point $(0,1-c)$ and that the pooling payoff graph $\widehat{v}_B^P$ is constructed from $\widehat{v}_B$ by restricting attention only to $\mu\geq k$, it is also evident that $V^{pp}(\cdot)$ is the concave envelope of the pooling payoff graphs, consistent with Proposition \ref{pre-per}.

\section{Application: Rating Design}

We study the design of an optimal certification, rating, or grading system, endowing the designer with a preference for the rating's \textit{nominal} realization. We focus on two channels by which the nominal grade or rating matters, beyond the information it reveals. First, the designer benefits when the rating is favorable, an important feature in a variety of settings. For example, financial analysts who issue optimistic ratings are more likely to be promoted than their more-accurate peers \citep{Hong/Kubik:2003}. Similarly, professors may be able to improve their course evaluations by assigning higher grades, which could improve their career prospects \citep{Johnson:2003}. Certification intermediaries may similarly wish to present favorable findings to preserve business relationships with clients. The existence of naive receivers, who take ratings or grades at face-value, implies that a favorable rating can produce a  stronger response than is warranted, based solely on its information content.\footnote{\citet{Inderst/Ottaviani:2012a, Inderst/Ottaviani:2012b} study models of financial advice with naive consumers. \citet{KOS2007} incorporate credulity into a model of strategic communication. Comparing individual and institutional investors, \citet{Malmendier/Shanthikumar:2007} and \citet{DeFranco/Lu/Vasvari:2007} provide evidence that individual investors tend to interpret analyst recommendations naively. In the context of hiring, \citet{HHH2023} find that some employers react credulously to changes in undergraduate GPA at the time of hiring.} Second, exaggeration may also impose costs. Financial analysts who issue misleading reports risk legal sanction. Certifiers who recommend unsuitable or dangerous products risk litigation. Professors may have a distaste for inflating the grades of undeserving students. More broadly, the designer may have an aversion to lying, particularly in a self-serving manner.

To fix ideas, consider a financial analyst who rates an  asset. The state $\omega\in\{b,g\}$ is the asset's quality, and the prior is $\mu_0=\Pr(\omega=g)$.  The analyst can assign two possible ratings to the asset, $s\in\{H,L\}$. Before learning the state, he designs an evaluation procedure or disclosure rule, $\pi(s|\omega)$, which specifies the probability that rating $s$ is issued given state $\omega$. We call $\pi(\cdot|\cdot)$  a rating policy. After setting the policy, the analyst privately observes quality and issues a public rating. 

The rating is observed by a continuum of investors with total mass 1, each of whom has a single infinitesimal unit of capital.\footnote{Allowing for investors who are not infinitesimal makes no difference to the analysis.}  Investors may be either sophisticated or naive. Sophisticated investors observe both the rating policy $\pi(\cdot|\cdot)$ and the realized rating when updating their beliefs. In contrast, naive investors interpret ratings literally: $s=H$ means $\omega=g$, and $s=L$ means $\omega=b$. The fraction of naive investors is $\nu\in(0,1)$. After observing the rating and updating beliefs, each investor $i\in[0,1]$ draws an independent outside option $\theta_i\sim F(\cdot)$, where $F$ is twice continuously differentiable with support $[0,1]$. Investors then decide whether to allocate capital to the asset or their outside options.\footnote{One can interpret $E[\theta]$ as the asset price, and $\theta_i-E[\theta]$ as investor $i$'s idiosyncratic benefit or loss from investing, e.g., tax motives, liquidity needs, or warm glow. To streamline the model and focus on the rating design, we abstract from endogenous price adjustments, though these could be incorporated within our framework. In some cases, (e.g., IPO) financial assets sell for fixed prices.} If an investor allocates his capital to a good asset ($\omega=g$), then his payoff is 1, to a bad asset, 0. If instead the investor takes the outside option, his payoff is $\theta_i$. Therefore, if sophisticated investors believe $\mu=\Pr(\omega=g)$, then aggregate investment is $(1-\nu) F(\mu)+\nu \mathcal{I}(s=H)$.

The analyst would like to increase investment in the asset. This incentivizes him to exaggerate the rating to exploit the naive investors. This incentive   is tempered by the possibility of sanctions and psychological costs. In particular, when the high rating ($s=H$) is assigned to a bad asset ($\omega=b$), the analyst pays a (expected) cost $k>0$. Normalizing the analyst's payoffs by $1/(1-\nu)$, we have the following interim payoffs,
\begin{align*}
    &\widehat{v}(\mu,H)=F(\mu)+(b-c(1-\mu))\\
    &\widehat{v}(\mu,L)=F(\mu),
\end{align*}
where $b\equiv \nu/(1-\nu)$ and $c\equiv k/(1-\nu)$. Formally, the payoffs with $b=0$ and $c>0$ correspond to a \citet{S1973} model, in which sender has a signaling action available whose cost depends on the state of the world. With $b>0$ and $c=0$, payoffs correspond to a money burning model, where one action is ``purely dissipative'' \citep{Austen-Smith/Banks:2001,Kartik:2007}. Note that identical payoffs would arise in the absence of naive investors ($\nu=0$) if $b$ was interpreted as a career concern.

With commitment, the analyst's optimal rating policy is determined by the interaction of two incentives. First, the concavity or convexity of $F(\cdot)$ determines whether the analyst wishes to reveal or conceal information from sophisticated investors. When $F(\cdot)$ is concave (convex), concealing (revealing) information increases the aggregate investment   from the sophisticated segment of the market. Second, the analyst would like to manipulate the naive investors to boost their investment, but he must be wary of a high rating in the bad state. In particular, when $b>c$, the analyst benefits (on net) from boosting naive investment in the bad state even if it means he incurs the cost. In contrast, when $c>b$, the boost to naive investment in the bad state is not enough to offset the cost, and the analyst has an incentive not to exaggerate.  Thus, the incentive to reveal or conceal information from sophisticated and naive investors may reinforce or oppose each other. Furthermore, these incentives determine the shape and position of the analyst's interim payoff graphs. In particular, both $\widehat{v}(\cdot,H)$ and $\widehat{v}(\cdot,L)$ have identical concavity to $F(\cdot)$, and their relative positions are determined by the sign of $b-c$. When $b>c$, the graph of $\widehat{v}(\cdot,H)$ lies above $\widehat{v}(\cdot,L)$. In the reverse case, the graph of $\widehat{v}(\cdot,H)$ crosses $\widehat{v}(\cdot,L)$ once, and this crossing is from below---see Figure \ref{conceff}.

 First, consider concave $F(\cdot)$. If $b>c$, then $\widehat{v}_H$ lies strictly above $\widehat{v}_L$. With concavity, it is optimal for the analyst to pool on the high rating. However, when $c>b$, the payoff graphs cross. Although each graph is concave when viewed in isolation, the crossing creates non-concavity around the point of intersection. Thus, the analyst benefits by joining the graphs near the crossing, as illustrated in Figure \ref{conceff}. Proposition \ref{spenceRA} follows from  inspection of Figure \ref{conceff}.
 
 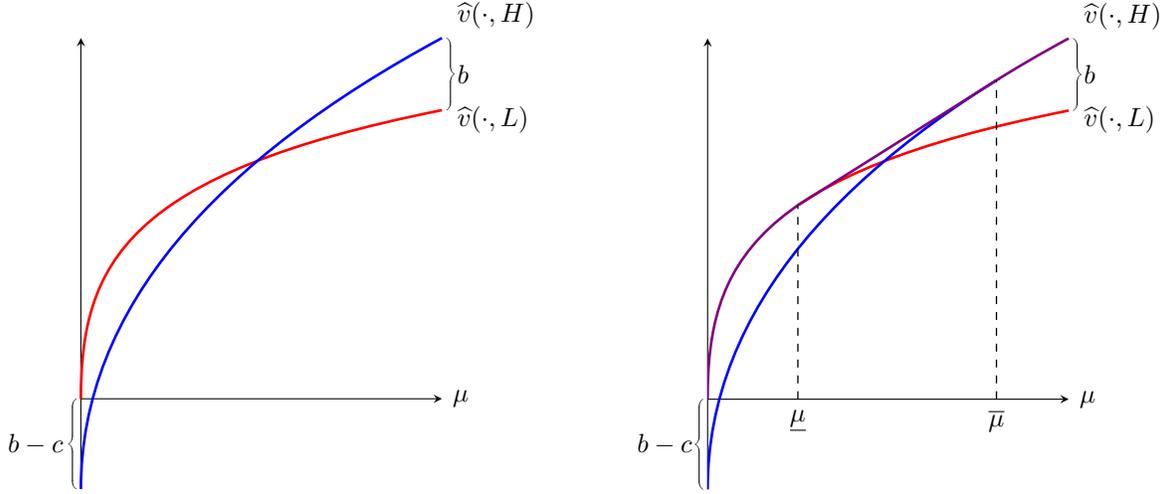
\begin{figure}
\begin{minipage}{0.5\textwidth}
%LEFT PANEL
\begin{center}
\begin{tikzpicture}[scale=0.6]

%AXES
\draw[-stealth] (0,-2) -- (0,8); 
\draw[-stealth] (0,0) -- (8,0); 

%AXES LABELS
\fill (8,0) node[right] {\footnotesize{$\mu$}};
\fill (8.1,8.5) node[right] {\footnotesize{$\widehat{v}(\cdot,H)$}};
\fill (8.1,6.2) node[right] {\footnotesize{$\widehat{v}(\cdot,L)$}};

\begin{scope}
%\clip(2,4.3) rectangle (8,8);
\draw[line width=1pt,red] (0,0) .. controls (0,3) and (1,5) .. (8,6.4);
\end{scope}

%\begin{scope}
%\clip(0,0) rectangle (2,4.35);
%\draw[line width=1pt,violet] (0,0) .. controls (0,3) and (1,5) .. (8,6.4);
%\end{scope}

%\begin{scope}
%\clip(0,-2) rectangle (6.5,7.07);
%\draw[line width=1pt,red] (0,-2) .. controls (0,2) and (3,5.3) .. (8,8);
%\end{scope}

%\begin{scope}
%\clip(6.1,7.07) rectangle (8,8);
%\draw[line width=1pt,violet] (0,-2) .. controls (0,2) and (3,5.3) .. (8,8);
%\end{scope}

\draw[line width=1pt,blue] (0,-2) .. controls (0,2) and (3,5.3) .. (8,8);

%  \draw[line width=1pt,violet] (2,4.3)--(6.4,7.07);

%\draw[line width=0.5pt,violet] (0,0.05) .. controls (0,3.05) and (1,5.05) .. (8,8.05);

%\fill[line width=1pt,color=violet,opacity=0.05] (0,0) .. controls (0,3) and (1,5) .. (8,8)--(8,7)--(0,-2);

%\draw[line width=1pt,violet] (0,0.1) .. controls (0,3.1) and (1,5.1) .. (8,8.);

\fill (8.1,7.2) node[right] {\footnotesize{$b$}};

  \draw[decorate,decoration={calligraphic brace,mirror,amplitude=3pt}] (8.1,6.4) -- (8.1, 8);

  \draw[decorate,decoration={calligraphic brace,mirror,amplitude=3pt}] (-0.1,0) -- (-0.1, -2);
  \fill (-0.15,-1) node[left] {\footnotesize{$b-c$}};
  
%  \draw[line width=0.5pt,black,dashed] (3,0)--(3,5.5);
%  \draw[line width=0.5pt,black,dashed] (0,-2)--(8,7);
%  \draw[line width=0.5pt,black,dashed] (0,0)--(8,7);
%\begin{scope}[transparency group,opacity=0.1]
% \fill[line width=1pt,color=violet] (0,-2) .. controls (0,2) and (3,5.3) .. (8,8)--(8,7)--(0,-2);
% \fill[line width=1pt,color=violet] (2,4.3)--(6.4,7.07)--(8,8)--(8,6.4);
%\fill[line width=1pt,violet] (0,0) .. controls (0,3) and (1,5) .. (8,6.4)--(8,6.4)--(0,-2); 
%   \end{scope}
%  \draw[line width=0.5pt,black,dashed] (2,0)--(2,4.3);
%    \draw[line width=0.5pt,black,dashed] (6.4,0)--(6.4,7.07);
%    \fill (6.4,0) node[below] {\footnotesize{$\overline{\mu}$}};
%        \fill (2,0) node[below] {\footnotesize{$\underline{\mu}$}};
\end{tikzpicture} \captionsetup{font=footnotesize}
%\captionof{figure}{Risk-averse sender/inefficient signaling \label{concin}}
\end{center}
\end{minipage}
\begin{minipage}{0.5\textwidth}
\begin{center}
\begin{tikzpicture}[scale=0.6]
\draw[-stealth] (0,-2) -- (0,8); 
\draw[-stealth] (0,0) -- (8,0); 

%AXES LABELS
\fill (8,0) node[right] {\footnotesize{$\mu$}};
\fill (8.1,8.5) node[right] {\footnotesize{$\widehat{v}(\cdot,H)$}};
\fill (8.1,6.2) node[right] {\footnotesize{$\widehat{v}(\cdot,L)$}};

\begin{scope}
\clip(2,4.3) rectangle (8,8);
\draw[line width=1pt,red] (0,0) .. controls (0,3) and (1,5) .. (8,6.4);
\end{scope}

\begin{scope}
\clip(0,0) rectangle (2,4.35);
\draw[line width=1pt,violet] (0,0) .. controls (0,3) and (1,5) .. (8,6.4);
\end{scope}

\begin{scope}
\clip(0,-2) rectangle (6.5,7.07);
\draw[line width=1pt,blue] (0,-2) .. controls (0,2) and (3,5.3) .. (8,8);
\end{scope}

\begin{scope}
\clip(6.1,7.07) rectangle (8,8);
\draw[line width=1pt,violet] (0,-2) .. controls (0,2) and (3,5.3) .. (8,8);
\end{scope}

%\draw[line width=1pt,red] (0,-2) .. controls (0,2) and (3,5.3) .. (8,8);

  \draw[line width=1pt,violet] (2,4.3)--(6.4,7.07);

%\draw[line width=0.5pt,violet] (0,0.05) .. controls (0,3.05) and (1,5.05) .. (8,8.05);

%\fill[line width=1pt,color=violet,opacity=0.05] (0,0) .. controls (0,3) and (1,5) .. (8,8)--(8,7)--(0,-2);

%\draw[line width=1pt,violet] (0,0.1) .. controls (0,3.1) and (1,5.1) .. (8,8.);

\fill (8.1,7.2) node[right] {\footnotesize{$b$}};

  \draw[decorate,decoration={calligraphic brace,mirror,amplitude=3pt}] (8.1,6.4) -- (8.1, 8);

  \draw[decorate,decoration={calligraphic brace,mirror,amplitude=3pt}] (-0.1,0) -- (-0.1, -2);
  \fill (-0.1,-1) node[left] {\footnotesize{$b-c$}};
  
%  \draw[line width=0.5pt,black,dashed] (3,0)--(3,5.5);
%  \draw[line width=0.5pt,black,dashed] (0,-2)--(8,7);
%  \draw[line width=0.5pt,black,dashed] (0,0)--(8,7);

% \begin{scope}[transparency group,opacity=0.1]
%  \fill[line width=1pt,color=violet] (0,-2) .. controls (0,2) and (3,5.3) .. (8,8)--(8,7)--(0,-2);
%  \fill[line width=1pt,color=violet] (2,4.3)--(6.4,7.07)--(8,8)--(8,6.4);
% \fill[line width=1pt,violet] (0,0) .. controls (0,3) and (1,5) .. (8,6.4)--(8,6.4)--(0,-2); 
%    \end{scope}
  \draw[line width=0.5pt,black,dashed] (2,0)--(2,4.3);
    \draw[line width=0.5pt,black,dashed] (6.4,0)--(6.4,7.07);
    \fill (6.4,0) node[below] {\footnotesize{$\overline{\mu}$}};
        \fill (2,0) node[below] {\footnotesize{$\underline{\mu}$}};
\end{tikzpicture}
\captionsetup{font=footnotesize}
\end{center}
\end{minipage}
\caption{$F(\cdot)$ concave, $c>b$\label{conceff}}
\end{figure}

\newpage
 \begin{prop}(Financial Ratings, Concave $F$).\label{spenceRA} In the model of financial ratings, if $F(\cdot)$ is concave and $c>b$, then belief thresholds $\underline{\mu},\overline{\mu}$ exist, with $0\leq\underline{\mu}<\overline{\mu}\leq 1$, such that,  
\begin{itemize} 
	\item if $\mu_0\leq\underline{\mu}$, then the optimal rating policy always assigns the low rating, $\mu_L=\mu_0$.
 
	\item if $\underline{\mu}<\mu_0<\overline{\mu}$, then the optimal rating policy induces beliefs $\{\mu_L=\underline{\mu},\mu_H=\overline{\mu}\}$. If $\underline{\mu}>0$, then the optimal rating understates asset value with positive probability, $\pi(L|g)>0$; if $\overline{\mu}<1$, then it overstates with positive probability, $\pi(H|b)>0$.
 
  	\item if $\mu_0\geq\overline{\mu}$,  then the optimal rating policy always assigns the high rating, $\mu_H=\mu_0$. 
\end{itemize}

\end{prop}

 To see the intuition, note that when $c>b$, the analyst would like to avoid exaggeration in the bad state and boost naive investment in the good state, pushing him toward a separating strategy. The incentive to separate is tempered by the concavity of $F(\cdot)$---expected investment by the sophisticated types is maximized when the rating is uninformative. At moderate priors, both good state and bad state are relatively likely, giving the strongest incentives for separation. In this case, the optimal rating policy is informative, though it generally misstates asset value with positive probability. When the prior is high, avoiding exaggeration in the bad state is less of a concern because the bad state is unlikely. Thus, the analyst pools on $H$. Similarly, when the prior is low, boosting naive investment in the good state is less of a concern because the good state is unlikely, and the analyst pools on $L$. 
 
 The logic of the previous paragraph suggests that the analyst's incentive to separate is driven by the existence of naive investors. Here, we show that when the cost of manipulation is not too large, naive investors are critical for the rating system to be informative.

 \begin{cor}\label{naive}(Naive Investors and Informativeness). Suppose $k\in(0,1)$ and $F(\cdot)$ is concave. At each prior belief $\mu_0\in(0,1)$, there exist thresholds $0<\nu_L(\mu_0)<\nu_M(\mu_0)<\nu_H(\mu_0)\leq 1$ such that (i)  the optimal rating is uninformative if $\nu<\nu_L(\mu_0)$, and (ii) the optimal rating is informative if $\nu\in(\nu_M(\mu_0),\nu_H(\mu_0))$.
 \end{cor}
 Thus, an increase in the fraction of naive investors from a low level ($\nu<\nu_L$) to a moderate level ($\nu_M<\nu<\nu_H$) increases the informativeness of the optimal rating system. 
 
 Finally, we point out that both understatement (arising when $\underline{\mu}>0$) and overstatement (arising when $\overline{\mu}<1$) are generic features of the optimal rating system. 
 
 \begin{remark}\label{underover}Consider a concave $F(\cdot)$ and any $(\mu_1,\mu_2)$ with $0<\mu_1<\mu_2<1$. There exist parameters $\{k,\nu\}$ such that $c>b$, and  $(\underline{\mu},\overline{\mu})=(\mu_1,\mu_2)$ in the optimal rating system described in Proposition \ref{spenceRA}.
 \end{remark}
 As we discuss further below, both of these findings (simultaneous over- and under-statement, Corollary \ref{naive}) are a direct consequence of the analyst's commitment power.

  Next, consider convex $F(\cdot)$. If $c>b$, then both  payoff graphs are convex, with $\widehat{v}_H$ crossing $\widehat{v}_L$ once from below. In this case, a truthful rating policy is optimal for the analyst, $\{\mu_H=1,\mu_L=0\}$. With convex $F(\cdot)$, full information maximizes aggregate investment from the sophisticated types. Furthermore, with $c>b$, it is too costly for the analyst to manipulate naive investors in the bad state. Both incentives reinforce each other, resulting in truthful ratings. With $c<b$, the incentives are more intricate. Here, the analyst would like to manipulate naive investors in the bad state, even if it means he pays the exaggeration cost. Thus, maximizing investment from the sophisticated types requires full separation, but from the naive types it requires full pooling on $H$. If a large fraction of investors are naive, pooling on $H$ may be better than separation (even though each payoff graph is convex). Furthermore, even full separation is not as straightforward as it may appear. With a fully separating strategy, the naive types invest in only one state, the one where the asset is rated $H$. Thus, it is advantageous for the analyst to commit to assign the high rating in the state that is (relatively) more likely. In other words, if the prior is high, then the analyst wants to commit to generate the high rating in the good state, resulting in a truthful rating. But, if the prior is low, then the analyst wants to commit to generate the high rating in the bad state, resulting in an \textit{inverted} rating.  Proposition \ref{prop:luie-1} follows from inspection of Figure \ref{conveff}.

  \begin{figure}
\begin{minipage}{0.5\textwidth}
%LEFT PANEL
\begin{tikzpicture}[scale=0.65]
\draw[-stealth] (0,0) -- (0,8); 
\draw[-stealth] (0,0) -- (8,0); 

%AXES LABELS
\fill (8,0) node[right] {\footnotesize{$\mu$}};

% \fill (3.8,0) node[below] {\footnotesize{$\overline{\mu}$}};

% \fill (2.45,0) node[below] {\footnotesize{$\underline{\mu}$}};

% \fill (3.2,0) node[below] {\footnotesize{$\hat{\mu}$}};

\draw[line width=1pt,blue] (0,0) .. controls (1,0) and (5,1) .. (8,5);

%\draw[line width=1pt,blue] (0,0) .. controls (0,3) and (1,5) .. (8,8);

\draw[line width=1pt,red] (0,2) .. controls (1,1.9) and (5,3.8) .. (8,8);

% \begin{scope}
% \clip(2.45,2.9) rectangle (3.8,3.8);
% \draw[line width=1pt,violet] (0,2) .. controls (1,1.9) and (5,3.8) .. (8,8);
% \end{scope}

%\draw[line width=0.5pt,violet] (0,0.05) .. controls (0,3.05) and (1,5.05) .. (8,8.05);

%\fill[line width=1pt,color=violet,opacity=0.1]  (0,-2) .. controls (1,-2) and (5,-1) .. (8,6.5)-- (8,8.1)-- (4.4,3.59) -- (0,0);

%\draw[line width=1pt,violet] (0,0.1) .. controls (0,3.1) and (1,5.1) .. (8,8.);

\fill (8.1,6.5) node[right] {\footnotesize{$b$}};
\fill (8.1,8.5) node[right] {\footnotesize{$\widehat{v}(\cdot,H)$}};
\fill (8.1,5) node[right] {\footnotesize{$\widehat{v}(\cdot,L)$}};

  \draw[decorate,decoration={calligraphic brace,mirror,amplitude=3pt}] (8.1,5.05) -- (8.1, 8);

  \draw[decorate,decoration={calligraphic brace,mirror,amplitude=3pt}] (-0.1, 2) -- (-0.1,0);
  \fill (-0.1,1) node[left] {\footnotesize{$b-c$}};

% \draw[line width=0.5pt, black, dashed] (2.55,2.95) -- (8,5); 
% \draw[line width=1pt, violet] (0,2) -- (2.45,2.9); 

% \draw[line width=0.5pt, black, dashed] (0,0) -- (3.8,3.75); 

 % \draw[line width=0.5pt, black, dashed] (2.5,0) -- (2.5,2.9); 

% \draw[line width=0.5pt, black, dashed] (3.8,0) -- (3.8,3.8); 
% \draw[line width=1pt, violet] (3.8,3.75) -- (8,8); 

% \draw[line width=0.5pt, black, dashed] (3.2,0) -- (3.2,3.2); 

%  \draw[line width=0.5pt,black,dashed] (3,0)--(3,5.5);
%  \draw[line width=0.5pt,black,dashed] (0,-2)--(8,7);
%  \draw[line width=0.5pt,black,dashed] (0,0)--(8,7);
% \begin{scope}[transparency group,opacity=0.1]
%  \fill[line width=1pt,color=violet] (0,0) .. controls (1,0) and (5,1) .. (8,6.05)--(8,8)--(0,0)--(8,8);
%  \fill[line width=1pt,color=violet] (0,-2) .. controls (1,-2) and (5,-1) .. (8,8)--(0,0)--(0,-2);
% %  \fill[line width=1pt,color=violet] (7.7,5.77)--(4.65,1)--(0,0)--(8,8);
%     \fill[line width=1pt,color=violet] (7.75,5.5)--(4.78,1.2)--(0,0)--(8,8);

% %\fill[line width=1pt,violet] (0,0) .. controls (0,3) and (1,5) .. (8,6.4)--(8,6.4)--(0,-2); 
%    \end{scope}
%   \draw[line width=1pt,color=black] (7.75,5.5)--(4.78,1.2);
%   \draw[line width=0.5pt,black](7.7,5.9)--(5,1.6);
\end{tikzpicture}
\end{minipage}
\begin{minipage}{0.5\textwidth}
\begin{tikzpicture}[scale=0.65]
\draw[-stealth] (0,0) -- (0,8); 
\draw[-stealth] (0,0) -- (8,0); 

%AXES LABELS
\fill (8,0) node[right] {\footnotesize{$\mu$}};

\fill (3.8,0) node[below] {\footnotesize{$\overline{\mu}$}};

\fill (2.45,0) node[below] {\footnotesize{$\underline{\mu}$}};

% \fill (3.2,0) node[below] {\footnotesize{$\hat{\mu}$}};

\draw[line width=1pt,blue] (0,0) .. controls (1,0) and (5,1) .. (8,5);

%\draw[line width=1pt,blue] (0,0) .. controls (0,3) and (1,5) .. (8,8);

\draw[line width=1pt,red] (0,2) .. controls (1,1.9) and (5,3.8) .. (8,8);

\begin{scope}
\clip(2.45,2.9) rectangle (3.8,3.8);
\draw[line width=1pt,violet] (0,2) .. controls (1,1.9) and (5,3.8) .. (8,8);
\end{scope}

%\draw[line width=0.5pt,violet] (0,0.05) .. controls (0,3.05) and (1,5.05) .. (8,8.05);

%\fill[line width=1pt,color=violet,opacity=0.1]  (0,-2) .. controls (1,-2) and (5,-1) .. (8,6.5)-- (8,8.1)-- (4.4,3.59) -- (0,0);

%\draw[line width=1pt,violet] (0,0.1) .. controls (0,3.1) and (1,5.1) .. (8,8.);

\fill (8.1,6.5) node[right] {\footnotesize{$b$}};
\fill (8.1,8.5) node[right] {\footnotesize{$\widehat{v}(\cdot,H)$}};
\fill (8.1,5) node[right] {\footnotesize{$\widehat{v}(\cdot,L)$}};

  \draw[decorate,decoration={calligraphic brace,mirror,amplitude=3pt}] (8.1,5.05) -- (8.1, 8);

  \draw[decorate,decoration={calligraphic brace,mirror,amplitude=3pt}] (-0.1, 2) -- (-0.1,0);
  \fill (-0.1,1) node[left] {\footnotesize{$b-c$}};
  \fill (-0.1,0) node[left] {\footnotesize{$\mu^T_L$}};

\fill (7.6,8) node[above] {\footnotesize{$\mu^T_H$}};

 % \draw[line width=1pt, green] (0,2) -- (8,8); 

  \fill (0,2.2) node[left] {\footnotesize{$\mu^I_H$}};

\fill (8.1,5.5) node[left] {\footnotesize{$\mu^I_L$}};

\draw[line width=0.5pt, black, dashed] (2.55,2.95) -- (8,5); 
\draw[line width=1pt, violet] (0,2) -- (2.45,2.9); 

\draw[line width=0.5pt, black, dashed] (0,0) -- (3.8,3.75); 

 \draw[line width=0.5pt, black, dashed] (2.5,0) -- (2.5,2.9); 

\draw[line width=0.5pt, black, dashed] (3.8,0) -- (3.8,3.8); 
\draw[line width=1pt, violet] (3.8,3.75) -- (8,8); 

% \draw[line width=0.5pt, black, dashed] (3.2,0) -- (3.2,3.2); 

%  \draw[line width=0.5pt,black,dashed] (3,0)--(3,5.5);
%  \draw[line width=0.5pt,black,dashed] (0,-2)--(8,7);
%  \draw[line width=0.5pt,black,dashed] (0,0)--(8,7);
% \begin{scope}[transparency group,opacity=0.1]
%  \fill[line width=1pt,color=violet] (0,0) .. controls (1,0) and (5,1) .. (8,6.05)--(8,8)--(0,0)--(8,8);
%  \fill[line width=1pt,color=violet] (0,-2) .. controls (1,-2) and (5,-1) .. (8,8)--(0,0)--(0,-2);
% %  \fill[line width=1pt,color=violet] (7.7,5.77)--(4.65,1)--(0,0)--(8,8);
%     \fill[line width=1pt,color=violet] (7.75,5.5)--(4.78,1.2)--(0,0)--(8,8);

% %\fill[line width=1pt,violet] (0,0) .. controls (0,3) and (1,5) .. (8,6.4)--(8,6.4)--(0,-2); 
%    \end{scope}
%   \draw[line width=1pt,color=black] (7.75,5.5)--(4.78,1.2);
%   \draw[line width=0.5pt,black](7.7,5.9)--(5,1.6);
\filldraw[color=blue, fill=white] (0,0) circle (3pt);
\filldraw[color=red, fill=white] (0,2) circle (3pt);

\filldraw[color=blue, fill=white] (8,5) circle (3pt);
\filldraw[color=red, fill=white] (8,8) circle (3pt);

\end{tikzpicture}
\end{minipage}
\captionsetup{font=footnotesize}
\caption{$F(\cdot)$ convex, $c<b$ \label{conveff}}
\end{figure}
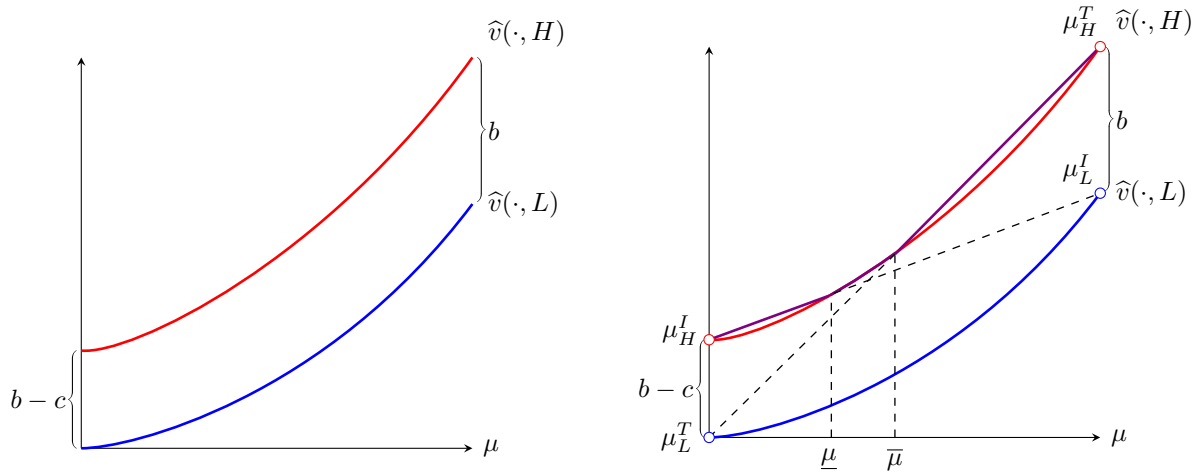

\begin{prop}\label{prop:luie-1}(Financial Ratings, Convex $F$). In the model of financial ratings, if $F(\cdot)$ is convex and $c<b$, then belief thresholds $\underline{\mu},\overline{\mu}$ exist with $0\leq\underline{\mu}\leq\overline{\mu}\leq 1$, such that
\begin{itemize}
\item if $\mu_0\leq\underline{\mu}$, then the optimal rating policy assigns the high rating to the bad asset and the low rating to the good asset, $\{\mu_L=1,\mu_H=0\}$.

\item if $\underline{\mu}< \mu_0< \overline{\mu}$, then the optimal rating policy always assigns the high rating, $\mu_H=\mu_0$.

\item if $\mu_0\geq \overline{\mu}$, then the optimal rating policy assigns the high rating to the good asset and the low rating to the bad asset, $\{\mu_L=0,\mu_H=1\}$.
\end{itemize}
\end{prop}

The inverted rating result is surprising. Intuitively,
with a convex $F(\cdot)$, sophisticated investors are likely to have good outside options. Provided that their fraction is sufficiently large, the analyst would like to maximize their aggregate investment by revealing the asset's quality. With a low prior belief, however, asset quality is likely to be low, and the sophisticated investors do not invest most of the time. To generate some investment when the asset is bad, the analyst can manipulate the naive types to invest. When the (expected) penalty for such manipulation is low, doing so is advantageous. Therefore, under these conditions, the optimal rating policy inverts the nominal meaning---sophisticated types infer that a highly rated asset is bad and avoid it, while naive types infer that a highly rated is good and invest.

The inverted rating result is consistent with 
empirical findings in the literature on financial analysts. Using data from the \textit{Global Research Analyst Settlement}, \citet{DeFranco/Lu/Vasvari:2007} find that in numerous instances analysts issued public guidance that was more positive than their private beliefs.
The response to these public announcements differed between individual investors (who are more likely to be naive) and (sophisticated) institutional investors. Crucially, they document that institutions  \textit{sold} the highly rated assets, while individual investors \textit{bought}. As we discuss further below, these findings are difficult to reconcile with a model in which analysts make public announcements without commitment. In such a model, sophisticated investors discount or ignore favorable announcements, but they do not interpret favorable announcements as bad news.

\paragraph{Formal Ratings vs. Informal Announcements} We conclude this section by highlighting the effects of the analyst's commitment power. In particular, suppose that instead of a formal rating conducted via a specific procedure that is directly observed by sophisticated investors, the analyst simply announces his rating, $s\in\{H,L\}$, after observing asset quality. In other words, consider the identical game without sender commitment power. Perfect Bayesian Equilibria of this game can be classified in the standard taxonomy. As is common in signaling games, in some cases pooling equilibria can be sustained by implausible off-path beliefs. We use the D1 refinement \citep{BanksS1987} to discipline the off-path beliefs, eliminating such equilibria.

\begin{remark}
 With the D1 refinement, the Perfect Bayesian Equilibrium without analyst commitment is unique. 
 \begin{itemize}
     \item If $b<c-1$, then the equilibrium is separating, with a high rating on the good asset and a low rating on the bad one, $\mu_H=1$ and $\mu_L=0$.
     \item If $b\in(c-1,c-F(\mu_0))$, then the equilibrium is semi-separating. The high rating is always assigned to the good asset, but both ratings are assigned with positive probability to the bad asset, $\mu_L=0$ and $c-F(\mu_H)=b$.
     \item If $b>c-F(\mu_0)$, then the equilibrium is pooling on $H$, with $\mu_L=0$ and $\mu_H=\mu_0$.

 \end{itemize}
 In all D1 equilibria, (i) a low rating fully reveals the bad state and a high rating is (weakly) good news, $\mu_L=0$ and $\mu_H\geq \mu_0$; (ii), the rating strategy is determined exclusively by $\{F(\mu_0),\nu,k\}$ (the overall shape of $F(\cdot)$ is irrelevant); (iii) rating informativeness is decreasing in the fraction of naive investors $\nu$ (or $b$).
 \end{remark}

  Intriguingly, all these three features are different when the analyst has commitment power. First, recall that with concave $F(\cdot)$, the optimal rating system not only overstates asset value, it may understate it as well. Indeed, whenever $\underline{\mu}>0$ and $\mu_0<\overline{\mu}$,  asset value is understated with positive probability ($\pi(L|g)>0$), and thus, the low rating does not fully reveal the bad state, $\mu_L>0$. Furthermore, with convex $F(\cdot)$, the analyst might use an inverted rating system in which the low rating reveals the good state and the high rating reveals the bad state $\mu_L=1>\mu_0>\mu_H=0$. Second, note that without commitment, the analyst's incentive to recommend $H$ in the bad state depends only on the level of investment he expects and on the exaggeration cost. With commitment, the global shape of $F(\cdot)$ determines whether the analyst has an incentive to reveal or conceal information, playing a critical role in the characterization. Third, recall from Corollary \ref{naive} that with concave $F(\cdot)$ and $c>b$, an increase in the fraction of naive investors can switch the optimal rating from uninformative to informative.

  \paragraph{Private Communication.} 

Consider the possibility that the analyst can privately communicate with sophisticated investors. Such private communications are heavily restricted by regulatory agencies. Nevertheless, understanding the consequences of such restrictions delivers useful insights. In particular, suppose that the analyst commits to a joint distribution of a rating $s\in\{H,L\}$ and a message $m\in M$, where $M$ is a large message space. Sophisticated investors observe the joint distribution, the rating, and the message and update their beliefs. Naive investors do not observe the message $m$: they invest if and only if the rating is high. Furthermore, assume that messages have no direct impact on the analyst's payoff: he incurs the cost only if he  exaggerates the asset's quality publicly, via the rating. By implication, the interim payoff graphs are unchanged, and this version of the model reduces to the extended commitment benchmark.

When $F(\cdot)$ is concave, both interim payoff graphs are concave. By Proposition \ref{infact2}, the analyst gains nothing from extended commitment. In contrast, when $F(\cdot)$ is convex, the extended commitment payoff is strictly above the \jenvelope{}, interpolating $(0,b-c)$ and $(1,1+b)$, both of which lie on $\widehat{v}(\cdot,H)$. In other words, the analyst always issues the high rating, and the naive investors always purchase the asset. Sophisticated investors learn the asset's true quality from the (private) message, and they invest if and only if the asset is good. Surprisingly, such private messages (extended commitment) may benefit the naive investors. In particular, when the prior belief is low $\mu_0\leq\underline{\mu}$, in the absence of such messages (signaling with commitment), the rating is inverted---naive investors buy when the asset is bad and do not buy when it is good. In contrast, with extended commitment, the naive investors buy \textit{both} when the asset is bad and when it is good. By investing in the good state, naive investors are better off.

\section{Conclusion}

In this paper, we study a standard signaling game, with the novel feature that the sender can commit to his strategy at the outset. We provide a geometric characterization of the sender's optimal payoff and strategy based on the topological join of the sender's interim payoff graphs. Extending our analysis to allow sender to commit to costless messages along with his action, we identify sender's gain and characterize environments in which it is zero. In addition, we characterize environments in which the sender must able to commit to actions in order to realize these gains. We provide novel insights in two applications: the design of adjudication procedures, and financial ratings.

Multiple directions for future research stand out. One direction is to extend the signaling with commitment environment 
to incorporate dynamics \citep{Kaya2009}, 
 imperfect observability of sender actions \citep{Jungbauer/Waldman:2023,Jungbauer/Waldman:2024}, receiver's private information \citep{HFT2002}, or some form of receiver commitment \citep{Whitmeyer:2024}.  Methodologically, the robustness literature uses the ``obedience approach'' based on Bayes Correlated Equilibrium, rather than the belief-based approach  that we use in our analysis. Understanding the connection between these approaches in signaling with commitment is a second direction for future research. This framework can also be applied to study questions in political economy, industrial organization, and finance.

\bibliography{fake}
\bibliographystyle{aer} 

\section{Appendix: Proofs}

\begin{proof}[Proof of Proposition \ref{BS1}] The proof follows \citet{KG2011}. It is presented only for completeness.

First, we show that a belief system induced by strategy $\Pi$ is Bayes-Plausible. 
\begin{align*}
	\sum_{s\in S}\tau(\mu_s)\mu_s(\omega)=\sum_{s\in S}\mu_0(\omega)\pi(s|\omega)=\mu_0(\omega)\sum_{s\in S}\pi(s|\omega)=\mu_0(\omega).
\end{align*}

Next,  we show that if a belief system is Bayes-Plausible, then it is induced by a strategy. Suppose belief system $\mathcal{B}\equiv\{\mu_s,\tau(\mu_s)\}_{s\in S}$ is Bayes-Plausible. Consider the following strategy:
\begin{align*}
	\pi(s|\omega)=\frac{\tau(\mu_s)\mu_s(\omega)}{\mu_0(\omega)},
\end{align*}
where the full support assumption ensures $\mu_0(\omega)>0$. Obviously, $\pi(s|\omega)\geq 0$. Moreover, Bayes-Plausibility implies $\sum_{s\in S}\tau(\mu_s)\mu_s(\omega)=\mu_0(\omega)$, and therefore $\sum_{s\in S}\pi(s|\omega)=1$. By construction, the proposed strategy induces belief system $\mathcal{B}$.
\end{proof}

\begin{proof}[Proof of Proposition \ref{sigcom}]\textit{Proof of part (i) ``if'' and part (ii)}. Suppose $(\mu_0,z)\in join(\widehat{v}_s)_{s\in S}$. By definition, $\{\lambda_s\}_{s\in S}$ exist such that $\lambda_s\geq 0$, $\sum_{s\in S}\lambda_s=1$, and $(\mu_0,z)=\sum_{s\in S}\lambda_s z_s$, where $z_s\in\widehat{v}_s$. From the definition of $\widehat{v}_s$, we have that for each $s\in S$, a $\mu_s\in\Delta(\Omega)$ exists such that $z_s=(\mu_s,\widehat{v}(\mu_s,s))$. Substituting, we have $(\mu_0,z)=\sum_{s\in S}\lambda_s(\mu_s,\widehat{v}(\mu_s,s))=(\sum_{s\in S}\lambda_s\mu_s,\sum_{s\in S}\lambda_s\widehat{v}(\mu_s,s))$. Thus,  belief system $\{\mu_s,\tau(\mu_s)=\lambda_s\}_{s\in S}$ is Bayes-Plausible ($\sum_{s\in S}\tau(\mu_s)\mu_s=\mu_0$) and attains payoff $z$ in sender's problem ($\sum_{s\in S}\tau(\mu_s)\widehat{v}(\mu_s,s)=z$).

\textit{Proof of part (i) ``only if.'' We show that if a Bayes-Plausible belief system $ \{\mu_s,\tau(\mu_s)\}_{s\in S}$ attains payoff $z$ in sender's problem, then $(\mu_0,z)\in join(\widehat{v}_s)_{s\in S}$.} Observe that $(\mu_0,z)=(\sum_{s\in S}\tau(\mu_s)\mu_s,\sum_{s\in S}\tau(\mu_s)\widehat{v}(\mu_s,s))=\sum_{s\in S}\tau(\mu_s)(\mu_s,\widehat{v}(\mu_s,s))$. Note that $(\mu_s,\widehat{v}(\mu_s,s))\in\widehat{v}_s$. Furthermore,  $\sum_{s\in S}\tau(\mu_s)=1$ and $\tau(\mu_s)\geq 0$. Therefore, with $\lambda_s=\tau(\mu_s)$, we have $(\mu_0,z)=\sum_{s\in S}\lambda_s z_s$, where $z_s=(\mu_s,\widehat{v}(\mu_s,s))\in\widehat{v}_s$. Therefore, $(\mu_0,z)\in join(\widehat{v}_s)_{s\in S}$.

\textit{Part (iii)} follows immediately from the definition of \jenvelope{}.\end{proof}

\subsection{Extended Commitment}
\begin{proof}[Proof of Proposition \ref{infodesign}] To prove the result, we show that in the information design benchmark, any combination of sender and receiver payoffs that can be attained  can also be attained with a public message, as in the extended commitment benchmark.

 Consider private communication. Sender's expected payoff is
\begin{align*}
\sum_{(m_\sigma,m_\rho,s,\omega)}\mu_0(\omega)\pi_{I}(m_\sigma,m_\rho,s|\omega)v(s,\widehat{a}(m_\rho,s),\omega),
\end{align*}
and receiver's expected payoff is 
\begin{align*}
\sum_{(m_\sigma,m_\rho,s,\omega)}\mu_0(\omega)\pi_{I}(m_\sigma,m_\rho,s|\omega)u(s,\widehat{a}(m_\rho,s),\omega),
\end{align*}
where the summation is over all $(m_\sigma,m_\rho,s,\omega)\in M_\sigma\times M_\rho\times S\times \Omega$ and $\widehat{a}(m_\rho,s)$ denotes the receiver's optimal response at $(m_\rho,s)$. Decompose the joint distribution $\pi_I(\cdot,\cdot,\cdot|\omega)$ to  into two parts: $\pi_I(m_\sigma,m_\rho,s|\omega)=\pi_\rho(m_\rho,s|\omega)\pi_\sigma(m_\sigma|m_\rho,s,\omega)$. Substituting, we have that sender's payoff is
\begin{align*}
\sum_{(m_\rho,s,\omega)}&\sum_{m_\sigma\in M_\sigma}\mu_0(\omega)\pi_\rho(m_\rho,s|\omega)\pi_\sigma(m_\sigma|m_\rho,s,\omega)v(s,\widehat{a}(m_\rho,s),\omega)=\\&\sum_{(m_\rho,s,\omega)}\mu_0(\omega)\pi_\rho(m_\rho,s|\omega)v(s,\widehat{a}(m_\rho,s),\omega),
\end{align*}
and receiver's is 
\begin{align*}
\sum_{(m_\rho,s,\omega)}\mu_0(\omega)\pi_\rho(m_\rho,s|\omega)u(s,\widehat{a}(m_\rho,s),\omega).
\end{align*}
Now, consider public communication with $M=M_\rho$ and a family of joint distributions $\Pi_E=\{\pi_\rho(m,s|\omega)|(m,s)\in M\times S,\omega\in\Omega\}$. Sender's expected payoff is $$\sum_{(m_\rho,s,\omega)}\mu_0(\omega)\pi_\rho(m_\rho,s|\omega)v(s,\widehat{a}(m_\rho,s),\omega)$$ and receiver's is $$\sum_{(m_\rho,s,\omega)}\mu_0(\omega)\pi_\rho(m_\rho,s|\omega)u(s,\widehat{a}(m_\rho,s),\omega),$$ which are evidently identical to the payoffs in the case of private communication.\end{proof}

\begin{proof}[Proof of Proposition \ref{inducible}] \noindent\textit{Part (i).} First, we show that if $\tau(\mu,s)$ is induced by a sender strategy, then the marginal distribution of the posterior belief is Bayes-Plausible.

Consider a sender strategy, $\overline{\Pi}=\{\pi(m,s|\omega)\,|\,m\in M,\,s\in S,\,\omega\in\Omega\}$. Any realization from this strategy $(m,s)$, generates a posterior belief about the state $\mu_{(m,s)}$ and arises with a certain probability $\hat{\tau}(m,s)$,
\begin{align*}
    \hat{\tau}(m,s)=\sum_{\omega\in \Omega}\pi(m,s|\omega)\mu_0(\omega)\quad\quad \Pr(\omega|m,s)=\mu_{(m,s)}(\omega)=\frac{\pi(m,s|\omega)\mu_0(\omega)}{\hat{\tau}(m,s)}.
\end{align*}
Consider a partition of the set of  realizations $(m,s)$ based on the posterior belief $\mu_{(m,s)}$ that they induce and the action $s$. Define a set $E(\mu,s)\equiv\{(m,s)|m\in M,\, \mu_{(m,s)}=\mu\}$, the set of all message/action pairs that generate posterior belief $\mu$ and have action $s$. By definition, joint probability $\tau(\mu,s)=\sum_{(m,s)\in E(\mu,s)}\hat{\tau}(m,s)$. Next, note that by the Law of Iterated Expectations (or direct substitution),
\begin{align}\label{m1}
    \sum_{(m,s)\in M\times S}\hat{\tau}(m,s)\mu_{(m,s)}=\mu_0.
\end{align}
Regrouping the sum according to the sets $E(\mu,s)$ yields
\begin{align}\label{m2}
   \sum_{(m,s)\in M\times S}\hat{\tau}(m,s)\mu_{(m,s)}&=\sum_{(\mu,s)\in\Delta(\Omega)\times S}\Big\{\sum_{(m,s)\in E(\mu,s)}\hat{\tau}(m,s)\mu_{(m,s)}\Big\}\notag \\
&=\sum_{(\mu,s)\in\Delta(\Omega)\times S}\Big\{\mu\ \sum_{(m,s)\in E(\mu,s)}\hat{\tau}(m,s)\Big\}\notag \\
&=\sum_{(\mu,s)\in\Delta(\Omega)\times S}\mu\ \tau(\mu,s).
\end{align}
Combining \eqref{m1} and \eqref{m2} yields
\begin{align*}
\sum_{(\mu,s)\in\Delta(\Omega)\times S}\tau(\mu,s)\mu=\mu_0.
\end{align*}
Decomposing $\tau(\mu,s)=\tau(\mu)\tau(s|\mu)$ in the above equation yields
\begin{align*}
    \sum_{(\mu,s)\in\Delta(\Omega)\times S}\tau(\mu)\tau(s|\mu)\mu=\sum_{\mu\in\Delta(\Omega)}\sum_{s\in S}\{\tau(\mu)\tau(s|\mu)\mu\}=\sum_{\mu\in\Delta(\Omega)}\tau(\mu)\mu=\mu_0.
\end{align*}

Next, we show that if $\tau(\mu)$ is Bayes-Plausible, then it can be induced by a strategy such that $\Pr(\omega|m,s)=\Pr(\omega|m)$.

Consider joint distribution $\tau(\mu,s)$ such that the marginal distribution $\tau(\mu)$ is Bayes-Plausible. From \citet{KG2011} Proposition 1, the sender can design a signal $\pi'(m|\omega)$ that induces $\tau(\mu)$ using the message space alone. Next, define a joint probability of messages and actions conditional on $\omega$ as $\pi'(m,s|\omega)=\Pr(m|\omega)\Pr(s|m,\omega)= \pi'(m|\omega)\tau(s|\mu=\mu_m)$, where $\mu_m$ is the posterior induced by message $m$ from signal $\pi'(m|\omega)$. Clearly, this construction induces $\tau(\mu,s)$.
Moreover, in this construction, $\Pr(s|m,\omega)=\tau(s|\mu=\mu_m)$, so that $\Pr(s|m,\omega)=Pr(s|m)$, and hence 
\begin{align*}
 \Pr(\omega|m,s)=\Pr(\omega|m)\frac{\Pr(s|m,\omega)}{\Pr(s|m)}=\Pr(\omega|m).
\end{align*}

\noindent\textit{Part (ii).} As part of the proof of part (i), we showed: if $\tau(\mu)$ is Bayes-Plausible, then it can be induced by a strategy such that $\Pr(\omega|m,s)=\Pr(\omega|m)$. This proves part (ii).
\end{proof}

\begin{lem}\label{infact4} The convex hull of the union of the sender's payoff graphs is equal to the convex hull of the join of the sender's payoff graphs, $con((\widehat{v}_s)_{s\in S})=con(join(\widehat{v}_s)_{s\in S})$. 
\end{lem}

\begin{proof}[Proof of Lemma \ref{infact4}] First, we show that $con(\widehat{v}_s)_{s\in S}\subseteq con(join(\widehat{v}_s)_{s\in S})$. This is immediate, since the union of the payoff graphs is a weak subset of their join (by definition).
	
Next, we show that $con(join(\widehat{v}_s)_{s\in S})\subseteq con(\widehat{v}_s)_{s\in S}$. Any point in $join(\widehat{v}_s)_{s\in S}$ can be expressed as a convex combination of some points belonging to the sender's payoff graphs. Thus, any point in $join(\widehat{v}_s)_{s\in S}$ is a convex combination of points in $\cup_{s\in S}\widehat{v}_s$. Therefore, any convex combination of points in $join(\widehat{v}_s)_{s\in S}$, is also a convex combination of points in $\cup_{s\in S}\widehat{v}_s$. The result follows.\end{proof}

\subsection{Pre-Persuasion}
\begin{lem}\label{lem:preper}
    Let $\mathcal{B}=(\mu_s,\tau(\mu_s))_{s\in S}$ be the belief system associated with a Perfect Bayesian Equilibrium of the signaling game that exists when the prior belief is $\mu$. Choose any belief $\mu_s\in \mathcal{B}$, such that $\tau(\mu_s)>0$ (i.e., choose any on-path action $s$ and the associated belief $\mu_s$). A pooling equilibrium exists at $(\mu_s,s)$.
\end{lem}

\begin{proof}[Proof of Lemma \ref{lem:preper}] To show existence of a pooling equilibrium on action $s$ at belief $\mu_s$, construct belief system $\mathcal{B}_{pool}$ from belief system $\mathcal{B}$ by setting the probability of each action $s'\neq s$ to 0 and setting the probability of action $s$ to 1, which obviously satisfies condition (1) of the definition. To verify condition (2), note that in the Perfect Bayesian Equilibrium represented by belief system $\mathcal{B}$, the probability of state $\omega$ given action $s$ is derived from Bayes' rule whenever $\tau(\mu_s)>0$. That is, if $\mu_s\in \mathcal{B}$ and $\tau(\mu_s)>0$, then
\begin{align*}
\mu_s(\omega)=\frac{\sum_{\omega\in\Omega}\mu(\omega)\pi(s|\omega)}{\tau(\mu_s)}.
\end{align*}
By implication, if $\mu_s(\omega)>0$ and $\tau(\mu_s)>0$, then $\mu(\omega)>0$ and $\pi(s|\omega)>0$. In other words, if state $\omega$ has strictly positive probability in the equilibrium belief associated with action $s$, $(\mu_s(\omega)>0)$, then $\omega$ must be in the support of the prior belief, and action $s$ must be chosen in state $\omega$ with strictly positive probability. It follows that action $s$ must be optimal for sender in state $\omega$, i.e., for any $s'\neq s$,
\begin{align}\label{poolIC}
    v(s,\widehat{a}(\mu_s,s),\omega)\geq v(s',\widehat{a}(\mu_{s'},s'),\omega).
\end{align}
This proves condition (2).    
\end{proof}

\begin{proof}[Proof of Proposition \ref{pre-per}]
For purposes of the proof, let 
\begin{align*}
    V^{pool}(\mu)&\equiv\sup\Big\{z\,|\,(\mu,z)\in con((\widehat{v}^P_{s})_{s\in S})\Big\}\\
    V^{pp}(\mu)&\equiv  \sup\Big\{z\,|\,(\mu,z)\in con(\widehat{v}_{eq})\Big\}.
\end{align*}
Our goal is to show that $V^{pool}(\mu)=V^{pp}(\mu)$. 

\noindent\textit{Step 1.} We show that $V^{pp}(\mu)\geq V^{pool}(\mu)$.  Consider any point $(\mu_0,z)\in con((\widehat{v}^P_s)_{s\in S})$. Because it lies in the convex hull of $\cup_{s\in S} \widehat{v}^P_s$, a joint distribution $\tau^{pool}\in \Delta(\Delta(\Omega)\times S)$ exists such that 
\begin{align*}
(\mu_0,z)=\sum_{(\mu,s)\in \text{sp}[\tau^{pool}]}\tau^{pool}(\mu,s)(\mu,\widehat{v}(\mu,s))
\end{align*}
and each $(\mu,\widehat{v}(\mu,s)\in \widehat{v}_s^P$. By definition, $(\mu,\widehat{v}(\mu,s))\in\widehat{v}_s^P$ if and only if $(\mu,s)\in P$, i.e., a pooling equilibrium exists at $(\mu,s)$. It follows that at each belief $\mu$ such that $(\mu,s)\in\text{sp}[\tau^{pool}]$, and therefore, $\widehat{v}(\mu,s)\leq \widehat{v}_{eq}(\mu)$. It follows that  
\begin{align*}
    z\leq z'\equiv\sum_{(\mu,s)\in \text{sp}[\tau^{pool}]}\tau^{pool}(\mu,s)\widehat{v}_{eq}(\mu).
\end{align*}
Furthermore, 
\begin{align*}
    (\mu_0,z')=\sum_{(\mu,s)\in \text{sp}[\tau^{pool}]}\tau^{pool}(\mu,s)(\mu_0,\widehat{v}_{eq}(\mu)).
\end{align*}
Thus, $(\mu_0,z')\in con(\widehat{v}_{eq})$. We have shown that, for each $(\mu,z)\in con((\widehat{v}^P_s)_{s\in S})$, there exists $(\mu,z')\in con(\widehat{v}_{eq})$ such that $z'\geq z$. By implication $V^{pp}(\mu)\geq V^{pool}(\mu)$.

\noindent\textit{Step 2.} We show that $V^{pool}(\mu)\geq V^{pp}(\mu)$. Consider any point in $(\mu,\widehat{v}_{eq}(\mu))\in\widehat{v}_{eq}$. The sender-preferred PBE that exists at $\mu$ induces a belief system $(\mu_s,\tau(\mu_s))_{s\in S}$, and $(\mu,\widehat{v}_{eq}(\mu))=\sum_{s\in S}\tau(\mu_s)(\mu_s,\widehat{v}(\mu_s,s))$. By Lemma \ref{lem:preper}, for each $s\in S$ such that $\tau(\mu_s)>0$, a pooling equilibrium exists at $(\mu_s,s)$, i.e., $(\mu_s,s)\in P$. Thus, for all such $s$, we have $(\mu_s,\widehat{v}(\mu_s,s))\in \widehat{v}_s^P$. We have shown that any point $(\mu,\widehat{v}_{eq}(\mu))$ in $\widehat{v}_{eq}$ can be written as a convex combination of points in the union of the pooling payoff graphs.  Therefore, $\widehat{v}_{eq}\subseteq con((\widehat{v}_s^P)_{s\in S})$. Taking the convex hull of both sets, we have $con(\widehat{v}_{eq})\subseteq con(con((\widehat{v}_s^P)_{s\in S}))=con((\widehat{v}_s^P)_{s\in S})$. By implication, $V^{pool}(\mu)\geq V^{pp}(\mu)$.\end{proof}

\subsection{Extended Commitment vs. Signaling With Commitment.}
\begin{proof}[Proof of Proposition \ref{infact2}] Part (i). From Lemma \ref{infact4}, we have that in the extended commitment benchmark, sender's optimal maximum payoff is the concave envelope of $V^{jo}(\cdot)$. The result follows from Jensen's inequality.

Part (ii). Suppose to the contrary that $\widehat{v}(\cdot,s)$ is concave for all $s\in S$, but $V^{jo}(\mu)\neq V^{co}(\mu)$ for some belief $\mu$. It must be that $V^{co}(\mu)$ places positive weight on more than one point from $\widehat v_s$ for some $s\in S$. In particular, suppose $(\mu,V^{co}(\mu))$ includes two points, $x_{1s}=(\mu_1,\widehat v(\mu_1,s))$ and $x_{2s}=(\mu_2,\widehat v(\mu_2,s))$ from the same graph $\widehat v_s$ in the convex combination with strictly positive weights $\lambda_1$ and $\lambda_2$, respectively. In the convex combination, replace these two points with $\hat \mu=\alpha \mu_1+(1-\alpha)\mu_2$, where $ \alpha=\lambda_1/(\lambda_1+\lambda_2)$, and the associated weight $\hat \lambda=\lambda_1+\lambda_2$. By concavity, $\hat\lambda\ \widehat v(\hat\mu,s)\geq \lambda_1\ \widehat v( \mu_1,s)+\lambda_2\ \widehat v( \mu_2,s)$. This replacement results in a new convex combination with a higher value of the sender's payoff, contradicting the definition on $V^{co}(\mu)$.
\end{proof}

\subsection{Extended Commitment vs. Pre-Persuasion}

\begin{proof}[Proof of Proposition \ref{prop:preper}]
 \noindent\textit{Part 1: ``if.''} 
 
 \noindent\textit{Step 1}. We show that $V^{pp}(\mu_0)\geq V^{co}(\mu_0)$. Consider a joint distribution that is optimal in the extended commitment benchmark for prior $\mu_0$, $\tau^{co}\in\Delta(G)$. By definition $(\mu_0,V^{co}(\mu_0))=\sum_{\text{sp}[\tau^{co}]}\tau^{co}(\mu,s)(\mu,\widehat{v}(\mu,s))$. By assumption, for each $(\mu,s)$ such that $\tau^{co}(\mu,s)>0$, a pooling equilibrium exists at $(\mu,s)$, i.e., $(\mu,s)\in P$. By definition, $(\mu,\widehat{v}(\mu,s))\in\widehat{v}_s^P$. Therefore, $(\mu_0,V^{co}(\mu_0))=\sum_{\text{sp}[\tau^{co}]}\tau^{co}(\mu,s)(\mu,\widehat{v}(\mu,s))\in con((\widehat{v}^P_s)_{s\in S})$. It follows that $V^{pp}(\mu_0)\geq V^{co}(\mu_0)$.  
 
 \noindent\textit{Step 2.} We show that $V^{co}(\mu_0)\geq V^{pp}(\mu_0)$. By definition, $\widehat{v}^P_s\subseteq \widehat{v}_s$, which implies $\cup_{s\in S}\widehat{v}^P_s\subseteq\cup_{s\in S}\widehat{v}_s$, and in turn, $con(\cup_{s\in S}\widehat{v}^P_s)\subseteq con(\cup_{s\in S}\widehat{v}_s)$.  The claim is immediate.

 \noindent Combining Steps 1 and 2, we have $V^{pp}(\mu_0)=V^{co}(\mu_0)$, completing the proof of Part 1.

\noindent\textit{Part 2: ``only if.''} Suppose that $V^{pp}(\mu_0)=V^{co}(\mu_0)$. Because $V^{pp}(\mu_0)$ is attained, a joint distribution $\tau^{pp}\in\Delta(G)$ exists, such that $(\mu_0,V^{pp}(\mu_0))=\sum_{\text{sp}[\tau^{pp}]}\tau^{pp}(\mu,s)(\mu,\widehat{v}(\mu,s))$ and $(\mu,\widehat{v}(\mu,s))\in \widehat{v}^P_s$, i.e., for each $(\mu,s)$ such that $\tau^{pp}(\mu,s)>0$ a pooling equilibrium exists at $(\mu,s)$. That is $\text{sp}[\tau^{pp}]\subseteq P$. Obviously, $(\mu,\widehat{v}(\mu,s))\in\widehat{v}_s$. Therefore, $(\mu_0,V^{pp}(\mu_0))\in con((\widehat{v}_s)_{s\in S})$. If, in addition $V^{pp}(\mu_0)=V^{co}(\mu_0)$, then $\tau^{pp}$ attains the extended commitment payoff. By implication $\tau^{pp}\in T^{co}$.\end{proof}

\subsection{Nominal Ratings}

\begin{proof}[Proof of Corollary \ref{naive}] 
\textit{Proof of Claim (i).} 

 \noindent \textit{Step 1. We show that there exists $\nu_L(\mu_0)>0$ such that, if $\nu<\nu_L(\mu_0)$, then for all $\mu\in[0,1]$,}
 \begin{align*}
     F(\mu)+\frac{\nu-k}{1-\nu}+\frac{k}{1-\nu}\mu<F'(\mu_0)(\mu-\mu_0)+F(\mu_0).
 \end{align*}
 Consider 
 \begin{align*}
 Q(\nu)&\equiv \min_{\mu\in[0,1]}F'(\mu_0)(\mu-\mu_0)+F(\mu_0)-F(\mu)-\Big(\frac{\nu-k}{1-\nu}+\frac{k}{1-\nu}\mu\Big),
 \end{align*}
 and let $\mu^*(\nu)$ be the argminimum. By the Maximum Theorem, $Q(\cdot)$ is continuous. Therefore, it is enough to show $Q(0)>0$.
\begin{align*}
    Q(0)=F'(\mu_0)(\mu^*(0)-\mu_0)+F(\mu_0)-F(\mu^*(0))+k(1-\mu^*(0)).
\end{align*}
First, suppose $\mu^*(0)<1$. Concavity of $F(\cdot)$ implies $F'(\mu_0)(\mu^*(0)-\mu_0)+F(\mu_0)-F(\mu^*(0))\geq 0$, and hence $Q(0)\geq k(1-\mu^*(0))>0$. Second, suppose $\mu^*(0)=1$. Strict concavity of $F(\cdot)$ and $\mu_0<1$ imply $F'(\mu_0)(\mu^*(0)-\mu_0)+F(\mu_0)-F(\mu^*(0))>0$, and hence $Q(0)>0$.

\noindent\textit{Step 2. We show that if $\nu<\nu_L(\mu_0)$, then any Bayes-Plausible belief system with $\tau(\mu_H)>0$ is worse for sender than pooling on $L$, i.e., $\mu_L=\mu_0$ and $\tau(\mu_L)=1$.}
Consider a prior belief $\mu_0\in(0,1)$ and a Bayes-Plausible distribution supported on posteriors $\{\mu_L,\mu_H\}$ with $\tau(\mu_H)>0$. 
 Consider the analyst's payoff of such a belief system,
 \begin{align*}
 G(\mu_L,\mu_H)\equiv\tau(\mu_L)\widehat{v}(\mu_L,L)+\tau(\mu_H)\widehat{v}(\mu_H,H).
 \end{align*}
 Using Step 1 and the concavity of $F(\cdot)$, we have 
 \begin{align*}
 &\widehat{v}(\mu_H,H)=F(\mu_H)+\frac{\nu-k}{1-\nu}+\frac{k}{1-\nu}\mu_H<F'(\mu_0)(\mu_H-\mu_0)+F(\mu_0),\\
 &\widehat{v}(\mu_L,L)=F(\mu_L)\leq F'(\mu_0)(\mu_L-\mu_0)+F(\mu_0).
 \end{align*}
 Using these bounds, along with $\tau(\mu_H)>0$, we have
 \begin{align*}
 G(\mu_L,\mu_H)<\tau(\mu_L)(F'(\mu_0)(\mu_L-\mu_0)+F(\mu_0))+\tau(\mu_H)(F'(\mu_0)(\mu_H-\mu_0)+F(\mu_0))=F(\mu_0),
 \end{align*}
 where the equality uses Bayes-Plausibility. Note that by pooling on $L$, the analyst's payoff is $\widehat{v}(\mu_0,L)=F(\mu_0)$. This completes the proof of Claim (i).
 
 \noindent\textit{Proof of Claim (ii).}

 Suppose $\mu_0\in(0,1)$. Consider a Bayes-Plausible belief system $\{\mu_L=\mu_0-\delta,\mu_H=\mu_0+\delta\}$, $\tau(\mu_H)=\tau(\mu_L)=1/2$. The payoff of such a belief system is 
 \begin{align*}
 S(\delta,\nu)\equiv G(\mu_0-\delta,\mu_0+\delta)=&\frac{1}{2}F(\mu_0-\delta)+\frac{1}{2}\left(F(\mu_0+\delta)+\frac{\nu-k+k\mu_0+k\delta}{1-\nu}\right).
 \end{align*}
The highest possible payoff from an uninformative belief system is 
\begin{align*}
 P(\nu)=\max\{\widehat{v}(\mu_0,L),\widehat{v}(\mu_0,H)\}=\max\left\{F(\mu_0),F(\mu_0)+\frac{\nu-k+k\mu_0}{1-\nu}\right\}.
\end{align*}
We show that there exists a $\nu^*(\mu_0)\in(0,1)$ and a $\delta^*>0$ such that $S(\delta^*,\nu^*(\mu_0))>P(\nu)$.  Let $\nu^*(\mu_0)=k(1-\mu_0)$. Note that $k<1\Rightarrow \nu^*(\mu_0)\in(0,1)$. Note further, 
 \begin{align*}
     \frac{\nu^*(\mu_0)-k+k\mu_0}{1-\nu}=0\Rightarrow P(\nu^*(\mu_0))=F(\mu_0).
 \end{align*}
It follows that
 \begin{align*}
 S(\delta,\nu^*(\mu_0))=&\frac{1}{2}F(\mu_0-\delta)+\frac{1}{2}\left(F(\mu_0+\delta)+\frac{\nu^*(\mu_0)-k+k(\mu_0+\delta)}{1-\nu^*(\mu_0)}\right)\\
 =&\frac{1}{2}F(\mu_0-\delta)+\frac{1}{2}\left(F(\mu_0+\delta)+\frac{k\delta}{1-\nu^*(\mu_0)}\right).
 \end{align*}
 Moreover, $S(0,\nu^*(\mu_0))=F(\mu_0)$, and
 \begin{align*}
     \frac{\partial S(\delta,\nu^*(\mu_0))}{\partial\delta}\Big|_{\delta=0}=\frac{k}{1-\nu^*(\mu_0))}>0
 \end{align*}
 Therefore, a $\delta^*>0$ exists, such that $S(\delta^*,\nu^*(\mu_0))>F(\mu_0)=P(\nu^*(\mu_0))$. Next, note that $S(\delta^*,\cdot)$ and $P(\cdot)$ are continuous at $\nu=\nu^*(\mu_0)<1$. By implication, a non-empty interval $(\nu_M(\mu_0),\nu_H(\mu_0))$ exists such that, if $\nu\in(\nu_M(\mu_0),\nu_H(\mu_0))$, then $S(\delta^*,\nu)>P(\nu)$. Thus, the proposed informative belief system improves on an uninformative one. Therefore, the optimal rating must be informative.\end{proof}
 \\
 \\
 \begin{proof}[Proof of Remark \ref{underover}] Consider $0<\mu_1<\mu_2<1$. From Figure \ref{conceff}, a sufficient condition for $\underline{\mu}=\mu_1$ and $\overline{\mu}=\mu_2$ consists of two parts. (1) The crossing of the payoff graphs, $\mu^\dag\equiv 1-b/c$, lies between $\mu_1$ and $\mu_2$, i.e., $\mu_1<\mu^\dag<\mu_2$. (2) the tangent line to $\widehat{v}(\cdot,L)$ at $\mu=\mu_1$ is also tangent to $\widehat{v}(\cdot,H)$ at $\mu=\mu_2$. Condition (2) is equivalent to
  \begin{align}\label{S2}
     \frac{d\widehat{v}(\mu,L)}{d\mu}\Big|_{\mu=\mu_1}=\frac{d\widehat{v}(\mu,H)}{d\mu}\Big|_{\mu=\mu_2}=\frac{\widehat{v}(\mu_2,H)-\widehat{v}(\mu_2,L)}{\mu_2-\mu_1}.
 \end{align}

 \textit{Step 1. We show that for any $0<\mu_1<\mu_2<1$, there exist parameters $b>0$ and $c>0$ such that (\ref{S2}) holds.} Substituting $\widehat{v}(\mu,L)$ and $\widehat{v}(\mu,H)$ into (\ref{S2}),
\begin{align*}
F'(\mu_1)=F'(\mu_2)+c=\frac{F(\mu_2)-F(\mu_1)+b-c(1-\mu_2)}{\mu_2-\mu_1}.         
\end{align*}
Solving, we have 
\begin{align}\label{S4}
    c&=F'(\mu_1)-F'(\mu_2)\\\nonumber
    b&=(F'(\mu_1)(1-\mu_1)+F(\mu_1))-(F'(\mu_2)(1-\mu_2)+F(\mu_2)).
    \end{align}

Because $F(\cdot)$ is concave, $c>0$. Furthermore, note that 
\begin{align*}
    \frac{d}{dx}\{F'(x)(1-x)+F(x)\}=F''(x)(1-x)<0,
\end{align*}
where the last inequality follows from concavity of $F(\cdot)$. By implication $b>0$. 

\noindent\textit{Step 2. We show that $\mu_1<\mu^\dag<\mu_2$ for $(b,c)$ in (\ref{S4}).} By routine simplification,
\begin{align*}
    1-\frac{b}{c}-\mu_1=\frac{\mu_2-\mu_1}{F'(\mu_1)-F'(\mu_2)}\left(\frac{F(\mu_2)-F(\mu_1)}{\mu_2-\mu_1}-F'(\mu_2)\right)>0,\\
    \mu_2-\left(1-\frac{b}{c}\right)=\frac{\mu_2-\mu_1}{F'(\mu_1)-F'(\mu_2)}\left(F'(\mu_1)-\frac{F(\mu_2)-F(\mu_1)}{\mu_2-\mu_1}\right)>0,
\end{align*}
where both inequalities follow from concavity of $F(\cdot)$. 

\noindent To complete the proof, note that the values of the parameters $(b,c)$ in (\ref{S4}) are equivalent to $k=c/(1+b)>0$ and $\nu=b/(1+b)\in(0,1)$.\end{proof}

\end{document}